\documentclass[reqno,11pt,a4paper]{amsart}

\usepackage{color}
\usepackage[numbers]{natbib}

\usepackage{hyperref}

\usepackage{amsmath, amsthm, amsfonts, amssymb}
\usepackage[all]{xy}

\usepackage{mathtools}
\mathtoolsset{showonlyrefs}
\numberwithin{equation}{section}
\usepackage[applemac]{inputenc}

\theoremstyle{plain}
\newtheorem{theorem}{Theorem}
\newtheorem{proposition}[theorem]{Proposition}
\newtheorem{lemma}[theorem]{Lemma}

\theoremstyle{definition}

\newtheorem{definition}[theorem]{Definition}

\theoremstyle{remark}
\newtheorem{remark}[theorem]{Remark}

\DeclareMathOperator{\im}{Im}
\def\Z{\mathbb{Z}}	
\def\C{\mathbb{C}}	
\def\R{\mathbb{R}}	
\def\S{\mathcal{S}}
\renewcommand{\leq}{\leqslant} 		
\renewcommand{\geq}{\geqslant}

\def\ud{\mathrm{d}}

\newcommand{\dev}{\partial}
\newcommand{\I}{\mathrm{i}}

\let\ker\relax
\DeclareMathOperator{\ker}{Ker}
\def\A{\mathcal{A}}
\def\hA{\hat{\A}}
\def\F{\mathcal{F}}
\def\hF{\hat{\F}}

\begin{document}

\title[Darboux-Getzler theorem for difference operators]{A Darboux-Getzler theorem for scalar difference Hamiltonian operators}
\author{Matteo Casati}
\author{Jing Ping Wang}
\address{School of Mathematics, Statistics and Actuarial Science\\University of Kent\\Canterbury CT2 9FS, United Kingdom}
\begin{abstract}
	In this paper we extend to the difference case the notion of Poisson-Lichnerowicz cohomology, an object encapsulating the building blocks for the theory of deformations of Hamiltonian operators. A local scalar difference Hamiltonian operator is a polynomial in the shift operator and its inverse, with coefficients in the algebra of difference functions, endowing the space of local functionals with the structure of a Lie algebra. Its Poisson-Lichnerowicz cohomology carries the information about the center, the symmetries and the admissible deformations of such algebra. The analogue notion for the differential case has been widely investigated: the first and most important result is the triviality of all but the lowest cohomology for first order Hamiltonian differential operators, due to Getzler. We study the Poisson-Lichnerowicz cohomology for the operator $K_0=\S-\S^{-1}$, which is the normal form for $(-1,1)$ order scalar difference Hamiltonian operators; we obtain the same result as Getzler did, namely $H^p(K_0)=0$ $\forall p>1$, and explicitly compute $H^0(K_0)$ and $H^1(K_0)$. We then apply our main result to the classification of lower order scalar Hamiltonian operators recently obtained by De Sole, Kac, Valeri and Wakimoto.
\end{abstract}
\maketitle

\tableofcontents

\section{Introduction}
Hamiltonian systems, both finite and infinite dimensional, are defined by two classes of objects: Poisson brackets, describing their underlying geometric structures, and Hamiltonian functions (or functionals), providing their dynamical content.

In the context of partial differential equations and of differential-difference equations, namely in the infinite dimensional settings, the notion of Poisson bracket is equivalently given in terms of so-called Hamiltonian operators.

Let us consider the class of evolutionary equation for a set of functions $u^i(x,t)$, $i\in\{1,\ldots,\ell\}$, of two (sets of) variables. Here $t$ is the time, or the parameter of the flow the equations define, and $x$ is the so-called \emph{independent (or space) variable}. We call $\ell$ the number of \emph{components} of the system.

An evolutionary system of differential equations 
\begin{equation}\label{eq:evo-intro}
\dev_t u^i=F^i(u,\dev u,\dev^2 u,\ldots)
\end{equation}
is said to be Hamiltonian for a (Hamiltonian) functional $H$ if it can be written in terms of Poisson brackets $\{\cdot,\cdot\}$ as
\begin{equation}\label{eq:evobracket-intro}
\dev_t u^i=\{u^i,H\}
\end{equation}
or, \emph{equivalently}, in terms of Hamiltonian structure $K$ as
\begin{equation}\label{eq:evooperat-intro}
\dev_tu^i=K\delta H.
\end{equation}
Here $\dev$ in \eqref{eq:evo-intro} is the partial derivation with respect to the space variables (that can be either one or several) acting on the dependent variables $u^j$. $\delta H$ is the variational derivative and the operator $K$ can be a differential or pseudo-differential operator, depending on the system under investigation.

The prototypical example is KdV equation
$$
\dev_t u = 6u u_x+u_{xxx}.
$$
It is Hamiltonian with respect to the operator $K=\dev_x$ and the local functional
$$
H=\int{u^3-\frac{u_x^2}{2}}.
$$

The study of Hamiltonian operators is particularly important in the theory of integrable systems and in deformation quantisation. It is well known, for instance, that Magri \cite{M78} introduced the concept of compatible pair of Poisson brackets (Hamiltonian structure) and related it to the complete integrability of systems of partial differential equations.

The notion of \emph{Poisson} (or Poisson-Lichnerowicz) \emph{cohomology} carries a lot of information about the properties of a Poisson bracket, or equivalently -- in the infinite dimensional setting -- of a Hamiltonian structure. On finite dimensional manifolds, Poisson brackets are identified with Poisson bivectors \cite{L77}. Such bivectors can be used to define a differential on the complex of multivectors, whose cohomology is the Poisson cohomology. It provides information about the center of the Poisson algebra (the Casimir function), its symmetries, and the compatible bivectors that can be defined on the same manifold.

Hamiltonian structures as the ones we have introduced can be interpreted as Poisson bivectors defined on some infinite dimensional manifold; this analogy is suggested by the fact that they define Poisson brackets with similar properties to the ones used for finite-dimensional systems; we then similarly define the complex of multivector fields and the Poisson cohomology.

For Poisson bivectors defined by differential operators, as the one we use for KdV equation, the main result has been proved by Getzler \cite{G01}.

He considers Hamiltonian operators of first order, for systems with one space variable; Dubrovin and Novikov have proved long ago \cite{DN83} that there always exist a system of coordinates for which such operators have the constant form
$$
K=K^{ij}\dev_x, \qquad\qquad K^{ij}=K^{ji}.
$$
for $i,j=1,\ldots,\ell$.
Getzler's theorem states that $H^p(K)=0$ for $p\geq 1$; in particular, the vanishing of the second and third cohomology group allows a complete classification of higher order compatible Hamiltonian operators, since they are all equal to $K$ after a generalised change of coordinates called a Miura transformation.

It is, for instance, well known that the second Hamiltonian structure for KdV equation
$$
K_2=4 u \dev_x+2u_x+\dev_x^3
$$
can be obtained by the first one $K=\dev_x$ after the Miura transformation
$$
u\mapsto v=u^2+ \I u_x.
$$
The scalar differential case depending on several independent variables has been addressed in \cite{CCS15,CCS18}: in this case the Poisson cohomology is infinite-dimensional, however its highly non-trivial third group imposes enough constraints to classify all the compatible Hamiltonian operators up to an arbitrary order.
\medskip
A natural extension from the continuous to the discrete setting is to study differential-difference systems. The basic example is the Volterra chain equation \cite{V31}; we have a function $u(n,t)$ of a lattice variable $n\in\Z$ and of the time, solution of the equation 
$$\dev_t \,u(n,t)=u(n,t)\left(u(n+1,t)-u(n-1,t)\right).$$
By denoting $u(n,t):=u_0=u, u(n+m,t):=u_m$ and introducing the shift operator $\S f(u,u_1,\ldots,u_n)=f(u_1,u_2,\ldots,u_{n+1})$ we can write the Volterra equation as
$$
\dev_t u=u(u_1-u_{-1}).
$$
This equation can be cast in Hamiltonian form defining the Hamiltonian \emph{difference} operator 
$$K=uu_1\S-u u_{-1}\S^{-1}$$
and the functional $H=\int u$, so that $\delta H=1$.

The foundations of calculus for difference operators have been developed by Kupershmidt \cite{Ku85}, and can be read alongside the better-known formal calculus of variations introduced by Gel'fand and Dikii \cite{GD75} for systems of PDEs. The notion of Poisson bivector is defined within this framework; however, we can define a tailored version of the $\theta$ formalism (on the lines of what Getzler did for the differential case) which allows to deal in a more efficient way with the complex of multivector fields. Moreover, De Sole, Kac, Valeri and Wakimoto have recently introduced the notion of multiplicative Poisson vertex algebras \cite{dSKVW18} which is yet another equivalent formulation and which can be effectively used for explicit computations.

The main purpose of this paper is the extension to the difference case of the notion of Poisson cohomology of a Hamiltonian structure. In this context,   
Hamiltonian structures are given by difference operators (we call them, in analogy with the differential case, \emph{local} operators), or ratios of difference operators \cite{CMW18-2}. Our principal result is the computation of the Poisson cohomology for a scalar, order $(-1,1)$ difference Hamiltonian operator. We obtain, although in a very different context and with more modern techniques, an analogue of Getzler's result.
\begin{theorem}\label{main}
	Let us consider the scalar difference Hamiltonian operator $K=\S-\S^{-1}$. Its Poisson cohomology is 
	$$H^p(K)=0\qquad\qquad \forall\, p>1.$$
	Moreover,
	\begin{align}
	H^0(K)&=\left\{\int \alpha+\int\beta u\; \Big|\; (\alpha,\beta)\in\C^2\right\}\\
	 H^1(K)&=\left\{\int\gamma\frac{\delta}{\delta u}\;\Big|\;\gamma\in\C\right\}
	\end{align}
\end{theorem}
The $0$-th cohomology group identifies the Casimir functionals of the Poisson bracket defined by $K$. The first cohomology is given by the evolutionary vector fields which are symmetries of the Hamiltonian structure $K$ but are not obtained as Hamiltonian flows. For the classification of compatible Hamiltonian structures, the important part of this result is the vanishing of the second and third cohomology: as we will discuss in Section \ref{sec:applic}, this replicates Getzler's result that any Hamiltonian structure compatible with $K$ is given by a Miura-type transformation of $K$ itself.

\medskip
The paper is organized as follows: in Section \ref{sec:formalism} we recall the formal calculus of variations for algebras of difference functions, introduced by Kupershmidt, and we define a ``difference'' $\theta$ formalism to describe the space of multivector fields, the Poisson bivector and its associated cohomology. In Section \ref{sec:Poisson} we specialise to the case of a single dependent variable and of a first order Hamiltonian structure: we compute its normal form and we prove our main Theorem. In Section \ref{sec:applic} we demonstrate a few applications of the main theorem to the classification of local Hamiltonian difference operators, observing that many of the ones described by De Sole, Kac, Valeri and Wakimoto \cite{dSKVW18} can be reduced to the constant $(-1,1)$ order form by a suitable change of coordinates. As a byproduct, we observe that all the compatible pairs of Hamiltonian operators listed in the classification produce the Volterra chain hierarchy. In Section \ref{sec:stretch} we generalise some of the results obtained for the Poisson cohomology to higher order constant Hamiltonian operators, obtaining an upper bound for the dimension of the cohomology groups. The triviality result for the compatible deformations of the operators does not extend to higher order operators.

\section{Functional variational calculus and deformations in the difference case}
\label{sec:formalism}

In this section we revise the formal calculus of variations in the difference-differential setting as originally laid out by Kupershmidt \cite{Ku85}, according to the more modern exposition of \cite{CMW18}. Moreover, we extend the so-called $\theta$ \emph{formalism} to the (difference) local multivector fields. The $\theta$ formalism for differential multivector fields was introduced by Getzler \cite{G01}. It is based on the observation that the algebra of multivector fields with its Schouten-Nijenhuis bracket can be defined in terms of an odd symplectic supermanifold \cite{Le77} and on Soloviev's definition of the Schouten bracket for a field theory \cite{so93}. It proved itself extremely useful to compute the Poisson cohomology for operators of differential type \cite{CPS16,CCS15,CCS18}. 

We introduce the basic notions of local multivectors,  Schouten-Nijenhuis brackets, and $\theta$ formalism in the general $\ell\in\mathbb{N}$ component case, which is the original Kupershmidt's setting. We will later specialise to the scalar (namely, $\ell=1$) case.

\subsection{Algebra of difference functions}

Let $(\mathcal{P}_\ell,\S)$ the algebra of polynomials over the field $\C$ in the variables $u^i_n$, $i\in\{1,\ldots,\ell\},\,n\in\Z$, endowed with an automorphism $\S$, defined by $\S u^i_n=\S u^i_{n+1}$ $\forall n$. It satisfies the following relation
\begin{equation}\label{eq:commder}
\S\frac{\dev}{\dev u^i_n}=\frac{\dev}{\dev u^i_{n+1}}\S,\quad i\in\{1,\ldots,\ell\},\,n\in\Z.
\end{equation}
\begin{definition}[\cite{dSKVW18}]
An \emph{algebra of difference functions} is a commutative associate unital algebra $(\A,\S)$, containing $\mathcal{P}_\ell$, endowed with commuting derivations $\frac{\dev}{\dev u^i_n}$ extending the ones in $\mathcal{P}_\ell$ and an automorphism $\S$ extending that on $\mathcal{P}_\ell$, such that the following two properties hold:
\begin{enumerate}
	\item $\frac{\dev f}{\dev u^i_n}=0$ for all but finitely many pairs $(i,n)$;
	\item Property \eqref{eq:commder} holds.
\end{enumerate}
\end{definition}
Moreover, we denote $\mathcal{C}=\{f\in\A\, |\, \S f=f\}$ the subalgebra of constants and $\bar{\mathcal{C}}=\{f\in\A\, |\, \frac{\dev f}{\dev u^i_n}=0\ \forall(i,n)\}$ the subalgebra of quasi-constants. Note that a non-constant quasi-constant is an element of $\A$ with only an explicit dependence (not through the variables $u^i_n$) on the lattice variable $n$.

For the \emph{scalar} difference functions we will focus on in this article, we will consider $\ell=1$ and $\bar{\mathcal{C}}=\mathcal{C}=\C$.

The elements of the quotient
\begin{equation}
\F = \frac{\A}{(\S-1)\A}
\end{equation}
are called local functionals. In particular, in $\F$ we have $\S f \sim f$.

We denote the projection map from $\A$ to $\F$ as a formal integral, which associates to $f\in\A$ the element $F:=\int f$ in $\F$.

The variational derivative of  a local functional $F=\int f$ is defined as 
\begin{equation} \label{varu}
\frac{\delta F}{\delta u^i} =\delta_{u^i}F:= \sum_{n\in\Z} \S^{-n} \frac{\partial f}{\partial u^i_n} .
\end{equation}
From the property (1) in the definition of the algebra of difference functions, this sum is always finite.

\begin{proposition} For any $f\in\A$,
\begin{equation} \label{varpa}
\frac{\delta}{\delta u^i} (\S-1)f =0 .  
\end{equation}
\end{proposition}
\begin{proof}
	This proposition is a special case of the deeper theorem $\ker\delta=\bar{\mathcal{C}}+(\S-1)\A$ (see \cite{Ku85,dSKVW18}, to which we refer for the complete proof). For our purposes, the proposition follows from the simple computation involving relation \eqref{eq:commder}
	\begin{equation}
	\begin{split}
	\frac{\delta}{\delta u^i}(\S-1)f&=\sum_{n\in \Z}\S^{-n}\frac{\dev }{\dev u^i_n}\S f-\sum_{n\in \Z}\S^{-n}\frac{\dev f}{\dev u^i_n}\\
	&=\sum_{n\in \Z}\S^{-n+1}\frac{\dev f}{\dev u^i_{n-1}}-\sum_{n\in \Z}\S^{-n}\frac{\dev f}{\dev u^i_n}\\
	&=\sum_{n'\in \Z}\S^{-n'}\frac{\dev f}{\dev u^i_{n'}}-\sum_{n\in \Z}\S^{-n}\frac{\dev f}{\dev u^i_n}=0.
	\end{split}
	\end{equation}
\end{proof}
From this proposition in particular it follows that the variational derivative of $F$ does not depend on the choice of the density $f$.

\subsection{Local multivectors}\label{sec:polyv}
A local $p$-vector $B$ is a linear $p$-alternating map from $\F$ to itself of the form
\begin{equation} \label{pvect}
B(I_1, \dots ,I_p) = \int B^{i_1 , \dots , i_p}_{n_1, \dots , n_p} \ \S^{n_1}  \left( \frac{\delta I_1}{\delta u^{i_1}} \right) \cdots  \S^{n_p} \left( \frac{\delta I_p}{\delta u^{i_p}} \right)
\end{equation}
where $B^{i_1 , \dots , i_p}_{n_1, \dots , n_p}  \in \A$, for arbitrary $I_1, \dots, I_p \in \F$. We denote the space of local $p$-vectors by $\Lambda^p \subset \mathrm{Alt}(\F^p, \F)$.

\medskip
We look at multivectors with lower $p$. A $0$-vector is a local functional. A local $1$-vector is a linear map
$$
\hat{X}(F)=\int X^i \frac{\delta F}{\delta u^i}.
$$
Indeed, from the following ``integration by parts'' rule
\begin{equation}\label{eq:intparts}
\int \left(\S^n f\right)\left(\S^m g\right)=\int f \S^{m-n}g
\end{equation}
it follows that
\begin{equation}\label{eq:1vfield}
\int B^i_n\S^n\frac{\delta F}{\delta u^i}=\int \left(\S^{-n}B^i_n\right)\frac{\delta F}{\delta u^i}:=\int X^i\frac{\delta F}{\delta u^i}.
\end{equation}
A difference \emph{evolutionary} vector field is a derivation of the algebra $\A$ that commutes with $\S$ and it is trivial on $\bar{\mathcal{C}}$. An evolutionary vector field of characteristic $X^i$, $i=\{1,\ldots,\ell\}$ is of the form
\begin{equation}\label{eq:evvfield}
\bar{X}(f)=\sum_{i=1}^N\sum_{n\in\Z}\S^n \left(X^i\right)\frac{\dev f}{\dev u^i_n}.
\end{equation}
There is a one-to-one correspondance between evolutionary vector fields and local 1-vectors. Any local 1-vector defines an evolutionary vector field of characteristic $X$ as in \eqref{eq:1vfield}. Conversely, for an evolutionary vector field $\bar{X}$ we have $\bar{X}(\S f)=\S(\bar{X}(f))$. From this it follows that the map
$$
\hat{X}\left(\int f\right):=\int \bar{X}(f),
$$
associating a local 1-vector $\hat{X}$ to any evolutionary vector field $\bar{X}$ is well defined on local functionals, since
$$
\int\bar{X}\left(f+(\S-1) g\right)=\int\left(\bar{X}(f)+(\S-1)\bar{X}(g)\right)=\int\bar{X}(f).
$$

\medskip
Local $2$-vector field can be identified with skewsymmetric operators, or brackets.

\begin{definition}
	A \emph{local (scalar) difference operator} $K$ is an element of $\A[\S,\S^{-1}]$, namely a finite sum of the form
	$$
	K=\sum_{n=M}^N a^{(n)} \S^n,
	$$
	with $a^{(n)}\in\A$, $M\leq N\in\Z$. We call the pair $(M,N)$ the \emph{order} of the operator.
\end{definition}

Let us consider a difference-operator valued matrix $K\in\mathrm{Mat}_\ell(\A[\S,\S^{-1}])$. It defines a bilinear operation, or a \emph{bracket}, among local functional $F,G$ by
\begin{equation}\label{eq:bracketdef}
\{F,G\}:=\int\sum_{i,j=1}^\ell \frac{\delta F}{\delta u^i} K^{ij}\frac{\delta G}{\delta u^j}.
\end{equation}
We say that the bracket is skewsymmetric if $\{F,G\}=-\{G,F\}$. Moreover, with a slight abuse of terminology we say that an operator is skewsymmetric if the bracket defined as in \eqref{eq:bracketdef} is skewsymmetric. Scalar local difference operator with this property (such as a Hamiltonian one) must be of order $(-N,N)$ for some $N>0$.

We identify a skewsymmetric bracket with a local 2-vector field by direct comparison with the definition in equation \eqref{pvect}. 

Indeed, let $K$ be a matrix of difference operators $K^{ij}=K^{ij}_{(n)}\S^n$ for $i,j=1,\ldots,\ell$ and $n\in\Z$. Then the bracket $\{F,G\}$ defined in \eqref{eq:bracketdef} corresponds to the bivector 
$$
B(F,G)=\int\sum_{i,j=1}^\ell\sum_{m,n\in\Z} B^{ij}_{m,n}\S^m\left(\frac{\delta F}{\delta u^i}\right)\S^n\left(\frac{\delta G}{\delta u^j}\right)
$$
with $B^{ij}_{0,n}=K^{ij}_{(n)}$ and $B^{ij}_{m,n}=0$ $\forall m\neq 0$.

We call a difference operator a \emph{Hamiltonian operator} if the bracket it defines is skewsymmetric and fulfils the Jacobi identity.

The property of being an Hamiltonian operator is usually expressed in terms of Fr\'echet derivative -- once we have defined the difference analogous of the Fr\'echet derivative for a differential operator, the property reads the same as the one given by Dorfman \cite{Do93}. 

The notion of multiplicative Poisson vertex algebras has been recently introduced by De Sole, Kac, Valeri and Wakimoto \cite{dSKVW18} as the algebraic structure underlying differential-difference Hamiltonian equations. For our purposes, we regard it as an equivalent definition of Hamiltonian operator. However, it is a convenient framework to perform explicit computations throughout our paper. To be self-contained, the complete definition and the exact terms of the equivalence are illustrated in Appendix \ref{app:mPVA}.

\subsection{The $\theta$ formalism}
The so-called $\theta$ formalism is an equivalent and computational-effective way to define local multivector fields in the theory of the formal calculus of variations and their Schouten bracket. In this section we provide a version tailored on the difference case, namely when the densities of local functionals are difference functions.

Let $\hA$ be the algebra of difference functions in the commutative variables $u^i_n$, $n\in\Z$ and of polynomials in the anticommutative variables 
$\theta_{i,n}$, i.e.,
\begin{equation}
\hA := \A[\left\{ \theta_{i,n},\, i\in\{1,\ldots,\ell\},\,n\in\Z \right\}] .
\end{equation}
$\hA$ is a graded algebra according to the super gradation $\deg_\theta$, by setting
\begin{equation} 
\deg_\theta u^i_n = 0 , \qquad \deg_\theta \theta_{i,n} = 1 .
\end{equation}
We denote $\hA^p$, the homogeneous components of $\hA$ with $\theta$-degree $p$. Clearly $\hA^0 =\A$.

The automorphism $\S$ is extended to $\hA$ by 
\begin{equation}
\S \theta_{i,n}=\theta_{i,n+1}.
\end{equation}
Moreover, $\ker(\S-1)$ on $\hA$ is $\bar{\mathcal{C}}$.

We denote by $\hF$ the quotient of $\hA$ by the subspace $(\S-1)\hA$, and by the integral operator $\int$ the projection map from $\hA$ to $\hF$. Since $\deg_\theta\S f=\deg_\theta f$, $\hF$ inherits the supergradation of $\hA$.  

Equation \eqref{varpa} holds on $\hA$ and, similarly,
\begin{equation}\label{commupatheta}
\S\frac{\dev}{\dev\theta_{i,n}}=\frac{\dev}{\dev\theta_{i,n+1}}\S.
\end{equation}
It follows that the $\theta$ variational derivative
\begin{equation} \label{varde}
\frac{\delta}{\delta \theta_i} =\delta_{\theta_i}:= \sum_{n\in\Z} \S^{-n} \frac{\partial }{\partial \theta_{i,n}} 
\end{equation}
satisfies
\begin{equation}\label{varderpa}
\frac{\delta}{\delta \theta_i} (\S-1) = 0.
\end{equation}
Hence both variational derivatives~\eqref{varpa} and~\eqref{varde} define  maps from $\hF$ to $\hA$.

\begin{proposition}
The space of local multi-vectors $\Lambda^p$ is isomorphic to $\hF^p$.
\end{proposition}
\begin{remark} 
A proof of this Proposition for the differential case is given in~\cite{lz11}. In the difference case the proof is simpler and it relies on a few theorems and lemmas originally proved by Kupershmidt \cite{Ku85}.
\end{remark}
\begin{proof}
For $p=0$, the isomorphism is trivial, since $\hF^0 = \F = \Lambda^0$. Let us assume instead that $p\geq1$. 
Given $B\in\hF^p$, and arbitrary $I_1, \dots, I_p \in \F$, let
\begin{equation}\label{defIota}
\iota (B)(I_1, \dots ,I_p) := \frac{\partial }{\partial \theta_{i_p,n_p}}\cdots \frac{\partial }{\partial \theta_{i_1,n_1}} B \cdot
\left(\S^{n_1} \frac{\delta I_1}{\delta u^{i_1}}\right) \cdots \left(\S^{n_p} \frac{\delta I_p}{\delta u^{i_p}}\right).
\end{equation}
Clearly $\iota(B)$ is an alternating map from $\hF^p$ to $\A$. Moreover it satisfies
\begin{equation}
\iota(\S B)(I_1, \dots , I_p) = \S\big(\iota(B)(I_1, \dots , I_p) \big)\sim\iota(B)(I_1, \dots , I_p).
\end{equation}

We can then define the map $\tilde{\iota}$ from $\hF^p$ to $\Lambda^p$ by 
\begin{equation}
\tilde{\iota}\left( \int B \right) := \int \iota(P).
\end{equation}
Surjectivity of $\tilde{\iota}$ is easy to see; indeed the local $p$-vector~\eqref{pvect} is the image through $\iota$ of
\begin{equation}
B = \frac1{p!}  B^{i_1 , \dots , i_p}_{n_1, \dots , n_p} \theta_{i_1,n_1} \cdots \theta_{i_p,n_p}.
\end{equation}
The injectivity of $\tilde{\iota}$ means that if
\begin{equation}
\iota (B)(I_1, \dots ,I_p)\sim 0 \qquad\forall\; (I_1,\dots,I_p),
\end{equation}
then $B\sim0$. For $p=1$ we have the stronger result that $\iota(B)\sim0$ implies $B=0$. Indeed, we have $\iota(X^i\theta_i)(F)=X^i\delta_{u^i}F\sim0$ for any $\delta F\in\A$, where $X^i\in\A$. This product is nondegenerate and it implies $X^i=0$ \cite[Lemma 17]{Ku85}. For $p\geq2$ we can always choose a representative $\tilde{B}$ of the form $\frac1p\theta_i\delta_{\theta_i}B$ for the element in $\hF^p$, which gives
\begin{equation}
\iota(B)(I_1, \dots ,I_p)\sim\frac{\delta I_1}{\delta u^{i_1}}\,\cdot\,\iota\left(\frac{\delta B}{\delta \theta_{i_1}}\right)(I_2, \dots ,I_p)\sim 0.
\end{equation}
Similarly to before, this means that $\iota(\delta_{\theta_i}B)=0$, which implies $\delta_{\theta_i}B=0$, hence $B\sim0$.
\end{proof}

\subsection{The Schouten-Nijenhuis bracket}
We define in the difference setting the so-called Schouten-Nijenhuis bracket for multivector fields. It is a bilinear map 
\begin{equation}
[,]:\hF^p \times \hF^q \to \hF^{p+q-1}
\end{equation}
defined as  
\begin{equation}\label{defsch}
[P , Q ] = \int\sum_{i=1}^N \left( \frac{\delta P}{\delta\theta_i} \frac{\delta Q}{\delta u^i} + (-1)^p \frac{\delta P}{\delta u^i} \frac{\delta Q}{\delta \theta_i} \right) .
\end{equation}

\begin{proposition}
	The (difference) Schouten-Nijenhuis bracket~\eqref{defsch} satisfies the graded symmetry 
\begin{equation}\label{eq:schsym}
[P,Q] = (-1)^{pq} [Q,P]
\end{equation}
and the graded Jacobi identity
\begin{equation}\label{eq:schjac}
(-1)^{pr} [[P,Q],R] + (-1)^{qp} [[Q,R],P] + (-1)^{rq} [[R,P],Q] =0 
\end{equation}
for arbitrary $P\in\hF^p$, $Q\in\hF^q$ and $r\in\hF^r$. Moreover, it extends the commutator of evolutionary vector fields for the case $p=q=1$.
\end{proposition}
\begin{proof}
Let us first prove that the Schouten bracket reduces to the usual commutator of vector fields in the case $p=q=1$. Two elements of $\hF^1$ are represented by $X^i\theta_i$ and $Y^j\theta_j$ with $X^i,\,Y^j\,\in\,\A$. Their Schouten bracket then reads
\begin{multline}
[X,Y]=\int \left(X^i\frac{(\delta Y^j\theta_j)}{\delta u^i}-Y^i\frac{(\delta X^j\theta_j)}{\delta u^i}\right)=\\
=\int\left(\left(\S^nX^i\right)\frac{\dev Y^j}{\dev u^i_n}-\left(\S^nY^i\right)\frac{\dev X^j}{\dev u^i_n}\right)\theta_j,
\end{multline}
corresponding to the commutator of two evolutionary vector fields. Indeed, we recall that
\begin{equation}
\hat{X}\left(\hat{Y}\left(\int f\right)\right)=\int X^i\frac{\delta}{\delta u^i}\left(\int Y^j\frac{\delta f}{\delta u^j}\right)=\int\S^mX^i\frac{\dev}{\dev u^i_m}\left(\S^n Y^j\frac{\dev f}{\dev u^j_n}\right).
\end{equation}
Taking the difference with the expression with exchanged $\hat{X}$ and $\hat{Y}$ we obtain
\begin{multline}
\int \S^mX^i\frac{\dev}{\dev u^i_m}\left(\S^nY^j\right)\frac{\dev f}{\dev u^j_n}-(X\leftrightarrow Y)=\\
\int\left( \S^{m-n}X^i\right)\frac{\dev Y^j}{\dev u^i_{m-n}}\S^{-n}\frac{\dev f}{\dev u^j_n}-(X\leftrightarrow Y)=\\
\int\left(\left( \S^{n}X^i\right)\frac{\dev Y^j}{\dev u^i_{n}}-\left( \S^{n}Y^i\right)\frac{\dev X^j}{\dev u^i_{n}}\right)\frac{\delta f}{\delta u^j}.
\end{multline}

\medskip

For $P\in\hF^p,\,Q\in\hF^q$, the skewsymmetry follows from the commutation rules of elements of $\hA$, namely $PQ=(-1)^{\deg_\theta P\deg_\theta Q}QP$ and from $\deg_\theta\delta_{\theta}=-1$. We have
\begin{equation}
[Q,P]=\int\sum_{i=1}^N\left((-1)^{p(q-1)}\frac{\delta P}{\delta u^i}\frac{\delta Q}{\delta \theta_i}+(-1)^{q+(p-1)q}\frac{\delta P}{\delta \theta_i}\frac{\delta Q}{\delta u^i}\right)=(-1)^{pq}[P,Q]
\end{equation}
The Jacobi identity is proved by a long but straightforward direct computation, in which the following properties are used:
\begin{equation}
\frac{\dev^2}{\dev u^i_m\dev u^j_n}=\frac{\dev^2}{\dev u^j_n\dev u^i_m},\quad\frac{\dev^2}{\dev u^i_m\dev \theta_{j,n}}=\frac{\dev^2}{\dev \theta_{j,n}\dev u^i_m},\quad\frac{\dev^2}{\dev \theta_{i,m}\dev \theta_{j,n}}=-\frac{\dev^2}{\dev \theta_{j,n}\dev \theta_{i,m}}.
\end{equation}
\end{proof}

\begin{definition}
	A bivector $P\in \hF^2$ is said to be a \emph{Poisson structure} if and only if $[P,P]=0$. Such a condition is sometimes called the \emph{Schouten identity} or the \emph{vanishing of the Schouten torsion}.
\end{definition}
\begin{proposition}\label{thm:OpBiv}
Let $P=\frac12\int P^{ij}_n\theta_i\theta_{j,n}\in\hF^2$ be a Poisson structure. Then the operator
\begin{equation}\label{eq:Hamfrombiv}
K^{ij}:=\left.\left(\frac{\delta P}{\delta \theta_i}\right)\right\vert_{\theta_{j,n}\to\S^n}
\end{equation}
is an Hamiltonian operator. Conversely, given a Hamiltonian difference operator $K^{ij}$, then the bivector
\begin{equation}\label{eq:bivfromHam}
P=\frac12\int \theta_iK^{ij}\theta_j
\end{equation}
is a Poisson structure. With the notation in \eqref{eq:bivfromHam} we mean that the operator $K^{ij}$ is applied to the variable $\theta_j$.
\end{proposition}
\begin{proof}
The equivalence between Hamiltonian operators and multiplicative Poisson vertex algebras (PVA) has been proved in \cite{dSKVW18}. The precise statement and some useful formulas are presented in Appendix \ref{app:mPVA}. In particular, any Hamiltonian difference operator defines a multiplicative PVA on the space of difference functions, defined on the generators by
\begin{equation}\label{eq:mPVAHam}
\{u^i{}_{\lambda}u^j\}=\left.K^{ji}\right\vert_{\S\to\lambda},
\end{equation}
In particular, the bracket~\eqref{eq:mPVAHam} is skewsymmetric and satisfies the PVA-Jacobi identity.

We first prove that the skewsymmetry is automatically granted by the definition~\eqref{eq:Hamfrombiv} and that~\eqref{eq:bivfromHam} is consistent with the skewsymmetry of $K$.
Let $P=\int P^{ij}_n\theta_i\theta_{j,n}$. Through formulae~\eqref{eq:Hamfrombiv} and \eqref{eq:mPVAHam} it defines a multiplicative $\lambda$ bracket $\{u^j{}_\lambda u^i\}=P^{ij}_n\lambda^n$. The skewsymmetry for this bracket (see Appendix \ref{app:mPVA-1}, property (4)) reads
\begin{equation}
\label{eq:skewExpl}
\S^{-n}P^{ji}_n\lambda^{-n}=-P^{ij}_n\lambda^n.
\end{equation}

Moreover, we have
\begin{equation}\label{eq:skewBiv}
\int P^{ij}_n\theta_i\theta_{j,n}=\int \left(\S^{-n}P^{ij}_n\right)\theta_{i,-n}\theta_j=-\int\left(\S^{-n}P^{ji}_n\right)\theta_i\theta_{j,-n}.
\end{equation}
Applying the map~\eqref{eq:Hamfrombiv} to the first and the third expressions of \eqref{eq:skewBiv} we obtain~\eqref{eq:skewExpl}. A similar computation shows that if the difference operator $K$ in \eqref{eq:bivfromHam} is not skewsymmetric, only its skewsymmetric part contributes to the definition of $P$.

We then prove that the condition $[P,P]=0$ is equivalent the PVA-Jacobi identity by a direct computation. The skewsymmetry property for a multiplicative $\lambda$-bracket imposes, in general, that
$$
\{u^i{}_{\lambda}u^j\}=P^{ij}_n\lambda^n-\S^{-n}P^{ji}_n\lambda^{-n},
$$
which, together with \eqref{eq:PVAJacobiHam}, gives an explicit form for the PVA-Jacobi identity of the general form
\begin{equation}\label{eq:PVAshort}
\sum_{m,n}A^{ijk}_{m,n}(\lambda^m\mu^n-\lambda^n\mu^m)=0,
\end{equation}
holding true $\forall\; i,j,k\in\{1,\ldots,\ell\}$, which is fulfilled if all the coefficients $A^{ijk}_{m,n}$ identically vanish.

On the other hand, the expression $[P,P]$ reads
\begin{multline}
2\int\frac{\delta P}{\delta\theta_k}\frac{\delta P}{\delta u^k}=\frac12\int\left(P^{kl}_m\left(\S^{-n}\frac{\dev P^{ij}_r}{\dev u^k_n}\right)\theta_{l,m}\theta_{i,-n}\theta_{j,r-n}\right.\\
\left.-\left(\S^{-m}P^{lk}_m\right)\left(\S^{-n}\frac{\dev P^{ij}_r}{\dev u^k_n}\right)\theta_{l,-m}\theta_{j,-n}\theta_{i,r-n}\right).
\end{multline}
The operator $\mathcal{N}=1/3\sum_k\theta_k\delta_{\theta_k}$ (introduced in the differential $\theta$ formalism in \cite{G01}) provides a normal form of the aforementioned expression, of the form
$$
\sum_{m,n}A^{ijk}_{m,n}\theta_k\theta_{i,m}\theta_{j,n},
$$
with the coefficients $A^{ijk}_{m,n}$ identical to the ones of \eqref{eq:PVAshort}. The explicit computations for the case $\ell=1$ are presented in Appendix \ref{app:PVAJacobi}.
\end{proof}

\subsection{The Poisson cohomology}\label{sec:Poisson}
A Poisson bivector defines, together with the Schouten-Nijenhuis bracket, a $\theta$-degree 1 differential $d_P:=[P,\cdot]$ on the space $\hF$.
It is obvious from the definition of Schouten bracket that $d_P\colon \hF^p\to\hF^{p+1}$. Moreover, by the graded Jacobi identity~\eqref{eq:schjac} we have
\begin{equation}
d_P^2 X=[P,[P,X]]=-[P,[P,X]]-[[P,P],X]=0.
\end{equation}
We can then define the cochain complex
$$
0\to\hF^0=\F\xrightarrow{d_P}\hF^1\xrightarrow{d_P}\hF^2\to\cdots
$$
whose cohomology is called the \emph{Poisson (or Poisson-Lichnerowicz)} cohomology. The Poisson cohomology is a graded space according to the $\theta$-degree $p$ of its elements (correspoding to local $p$-vector fields).
\begin{equation}
H(d_P,\hF)=\bigoplus_{p\geq 0}H^p(d_P,\hF)=\frac{\ker d_P\colon\hF^p\to\hF^{p+1}}{\im d_P\colon\hF^{p-1}\to\hF^p}.
\end{equation}
The interpretation of the lower cohomology groups is well-known and it is the analogue of the one in the finite dimensional case. The elements of $H^0$ are the Casimirs functionals of the bracket defined by $P$ -- elements $F$ whose variational derivative $\delta_u F$ is in the kernel of the Hamiltonian operator.

The elements of $H^1$ are evolutionary vector fields corresponding to symmetries of the bracket which are not Hamiltonian.

The cohomology groups $H^2$ and $H^3$ play a central role in the theory of deformation of the Hamiltonian structures, where the first one classifies their infinitesimal compatible deformations and the second one the obstructions to extend the deformation from infinitesimal to finite.

\medskip

Given a Poisson bivector $P_0$, an \emph{infinitesimal compatible deformation} of $P_0$ is a bivector $P_1$ such that $[P_0+\epsilon P_1,P_0+\epsilon P_1]=O(\epsilon^2)$. The compatibility is hence equivalent to $d_{P_0}P_1=0$. If $P_1$ is itself a Poisson bivector, namely $[P_1,P_1]=0$, than the deformation is said to be \emph{finite} and, in particular, $P_0$ and $P_1$ form a biHamiltonian pair.

We call an infinitesimal deformation \emph{trivial} when $P_1$ can be obtained from $P_0$ by the action of a vector field; this correspond to obtaining a new multiplicative $\lambda$ bracket $\{u_\lambda u\}'$ by introducing new coordinates $v=u+\epsilon f(u,u_1,\ldots,u_n)$, taking the order $\epsilon$ term in $\{v_\lambda v\}$ and expressing it in terms of the new variables. Deformed bracket of this form are always compatible, since they correspond to $P_1=d_P X$ for an evolutionary vector field $X$ (of characteristic $f$). Hence, the second cohomology group $H^2$ classifies the \emph{nontrivial} infinitesimal compatible deformations of a given Poisson bivector.

If $P_0+\epsilon P_1$ is not a Poisson bivector, the deformation may be extended to the order $\epsilon^2$ by adding a further bivector $P_2$, such that
$$
[P_0+\epsilon P_1+\epsilon^2 P_2,P_0+\epsilon P_1+\epsilon^2 P_2]=\epsilon^2\left([P_1,P_1]+2[P_0,P_2]\right)+O(\epsilon^3).
$$
$[P_1,P_1]$, when not vanishing, is a three-vector that -- as can be shown using the Jacobi identity for the Schouten bracket -- lies in the kernel of $d_{P_0}$. To let the $\epsilon^2$ term in the expansion vanish, there must exist $P_2$ such that $-2d_{P_0}P_2$ is equal to $[P_1,P_1]$. Such a bivector $P_2$ always exists if the third cohomology group $H^3$ is trivial, i.e. all the 3-vectors in the kernel of $d_{P_0}$ are of the form $d_{P_0}B$ for a bivector $B$. Note that the converse is not necessarily true, namely that there may exist a bivector $P_2$ as above, for \emph{a particular infinitesimal deformation} $P_1$, even in the case of nontrivial $H^3$ (so in the presence of obstructions to the extension of generic deformations).

In this paper we prove, in particular, that the second and third cohomology group for a scalar, order $(-1,1)$, Hamiltonian difference operator are trivial -- hence that there exist a change of coordinates, or a sequence of change of coordinates, for which the deformation vanishes. Moreover, we explicitly compute the zeroeth and the first cohomology groups.
\section{Poisson cohomology for a scalar Hamiltonian operator}
Let us consider an Hamiltonian operator of order $(-1,1)$. As proved in \cite{dSKVW18}, the skewsymmetry condition implies that it is of the form
\begin{equation}
K=f(u,u_1,u_{-1},\ldots)\S-\S^{-1}f(u,u_1,u_{-1},\ldots),
\end{equation}
while from the Jacobi identity it follows that $f=f(u,u_1)$ and
$$
(\S f)\frac{\dev f}{\dev u_1}=f\S\left(\frac{\dev f}{\dev u}\right).
$$
The condition is equivalent to
$$
\frac{\dev f}{\dev u_1}\Big/f=\S\left(\frac{\dev f}{\dev u}\Big/f\right),
$$
for which we note that the LHS depends on $u$ and $u_1$, while the RHS depends on $(u_1,u_2)$. This means that
$$
\frac{\dev}{\dev u_1}\log f=a(u_1)=\S \frac{\dev }{\dev u} \log f=\S a(u),
$$
namley $\log f=A(u)+A(u_1)+c$ and
$$
f=g(u)g(u_1)
$$
for some function $g$ of single variable. If $g(u)$ is not vanishing, the change of coordinates $v=\int^u\frac{1}{g(s)}ds$ brings the difference operator to the normal form
\begin{equation}\label{eq:normalform-op}
K_0=\S-\S^{-1}.
\end{equation}
This can be verified by a direct computation using the PVA formalism with the formula \eqref{eq:master}. The generic $(-1,1)$ order Hamiltonian operator corresponds to the $\lambda$ bracket $\{u_{\lambda}u\}=g(u)g(u_1)\lambda-g(u_{-1})g(u)\lambda^{-1}$. We have
\begin{multline}\label{eq:chNF}
\{v_{\lambda}v\}=\left\{\int \frac{1}{g(u)}{}_\lambda\int \frac{1}{g(u)}\right\}=\\
\frac{1}{g(u)}g(u)g(u_1)\left(\S\frac{1}{g(u)}\right)\lambda-\frac{1}{g(u)}g(u_{-1})g(u)\left(\S^{-1}\frac{1}{g(u)}\right)\lambda^{-1}=\\
\lambda-\lambda^{-1}.
\end{multline}
According to the identification \eqref{eq:bivfromHam}, the corresponding Poisson bivector in $\hF$ is
\begin{equation}\label{eq:normalform}
P=\int\theta\theta_1.
\end{equation}

The first basic building block of our computation is observing the exactness of the short sequence
\begin{equation}\label{eq:shortexact}
\xymatrix{
	0\ar[r]&\hA/\mathcal{C}\ar[r]^{\S-1}&\hA\ar[r]^\int&\hF\ar[r]&0.
	}
\end{equation}
This is obvious because the kernel of $(\S-1)$ is indeed the subalgebra of constants, and $\hF$ is by definition the quotient space $\hA/(\S-1)\hA$. From now on we will deal with the explicit scalar case \eqref{eq:normalform}, for which $\ell=1$ and $\mathcal{C}=\C$.

To compute the cohomology $H(\hF,d_{P})$ we introduce an auxiliary complex $(\hA,D_P)$ such that the following diagram commutes:
\begin{equation}
\xymatrix{
	0\ar[r]&\A\ar[d]^\int\ar[r]^{D_{P}}&\hA^1\ar[d]^\int\ar[r]^{D_{P}}&\hA^2\ar[d]^{\int}\ar[r]^{D_{P}}&\cdots\\
	0\ar[r]&\F\ar[r]^{d_{P}}&\hF^1\ar[r]^{d_{P}}&\hF^2\ar[r]^{d_{P}}&\cdots
}
\end{equation}

We will then exploit the long exact sequence in cohomology induced by \eqref{eq:shortexact}, namely
\begin{equation}\label{eq:longexact}
\xymatrix{
	\cdots\ar[r]&H^p(\hA/\C)\ar[r]&H^p(\hA)\ar[r]&H^p(\hF)\ar[r]&H^{p+1}(\hA/\C)\ar[r]&\cdots,
	}
\end{equation}
to compute the cohomology $H(\hF,d_P)$.

\subsection{The differential on $\hA$}

Given an element $P \in \hF^2$ we define the following differential operator on $\hA$
\begin{equation}\label{eq:defDP}
D_P := \sum_n \left(  \left(\S^n \frac{\delta P}{\delta \theta} \right) \frac{\partial }{\partial u_n}  + \left(\S^n \frac{\delta P}{\delta u} \right) \frac{\partial }{\partial \theta_n}  \right) .
\end{equation}
Since $[D_P, \S ] =0$, the operator $D_P$ descends to an operator on $\hF$ which is given by the adjoint action $ad_P = [P, \cdot]$ of $P$ on $\hF$ via the Schouten-Nijenhuis bracket, i.e.,
\begin{equation}\label{eq:fromDtod}
ad_P \left( \int Q \right)= \int D_P (Q),
\end{equation}
for $Q \in \hA$. If $P$ is a Poisson bivector, $ad_P=d_P$ and $d_P^2=0$.

For $P$ as in \eqref{eq:normalform} we have
\begin{equation}
D_P=\sum_n\left(\theta_{n+1}-\theta_{n-1}\right)\frac{\dev}{\dev u_n}
\end{equation} 
for which we have $D_P^2=0$. Indeed
\begin{equation}
D_P^2 f=\sum_{n,m}\left(\theta_{m+1}-\theta_{m-1}\right)\left(\theta_{n+1}-\theta_{n-1}\right)\frac{\dev^2 f}{\dev u_m\dev u_n},
\end{equation}
which is the product of a skewsymmetric (the product of $\theta$'s) and a symmetric (the second derivative) terms in $(n,m)$. Finally, a simple computation shows that $D_P$, by the map \eqref{eq:fromDtod}, gives the differential $d_PQ=\int\delta_\theta P\delta_uQ$.

\begin{remark}
	The operator $D_P$ for $P=\int\theta\theta_1$ can be also obtained as the prolongation of the vector field with characteristic $\theta_1-\theta_{-1}$.
\end{remark}
From $D_P\colon\hA^p\to\hA^{p+1}$ and $D_P^2=0$ it follows that we can introduce the Poisson-Lichnerowicz complex
\begin{equation}
\xymatrix{
	0\ar[r]&\hA^0=\A\ar[r]^{D_P}&\hA^1\ar[r]^{D_P}&\hA^2\ar[r]^{D_P}&\hA^3\ar[r]^{D_P}&\cdots.
}
\end{equation}
Its cohomology is, by the standard definition,
\begin{equation}
H(D_P,\hA)=\bigoplus_{p\geq 0}H^p(D_P,\hA)=\frac{\ker D_P\colon\hA^p\to\hA^{p+1}}{\im D_P\colon\hA^{p-1}\to\hA^p}.
\end{equation}
\begin{lemma}\label{poincare}
	$H(D_P,\hA)=\C[\theta,\theta_1]$. Since $\theta$ and $\theta_1$ are Grassmann variables, this means in particular that the cohomology is $4$-dimensional and it is generated as a vector space by $\langle 1,\theta,\theta_1,\theta\theta_1\rangle$.
\end{lemma}
\begin{proof}
The basic idea is that we can regard $D_P=(\theta_{n+1}-\theta_{n-1})\dev_{u_n}$ as a De Rham-type differential on the space $\hA$, identifying $(\theta_{n+1}-\theta_{n-1})$ with $d u_n$. On a topologically trivial space like $\hA$, the De Rham cohomology is concentrated in the constants in $H^0$, namely elements of $\C[\{\theta_n\}]$ without terms in the ideal generated by $(\theta_n-\theta_{n-2})$ for all $n$. In fact, any element containing a $\theta_{\bar{n}}$ variable is in the same cohomology class of the one with $\theta_{\bar{n}-2}$, since we have -- possibily with a sign given by the order of permutation of $\theta$ variables --
\begin{multline}
\theta_{n_1}\theta_{n_2}\theta_{\bar{n}}\cdots\theta_{n_p}=\theta_{n_1}\theta_{n_2}\theta_{\bar{n}-2}\dots\theta_{n_p}+(\theta_{\bar{n}}-\theta_{\bar{n}-2})\theta_{n_1}\theta_{n_2}\cdots\theta_{n_p}=\\
\theta_{n_1}\theta_{n_2}\theta_{\bar{n}-2}\dots\theta_{n_p}+D_P\left(u_{\bar{n}-1}\theta_{n_1}\theta_{n_2}\cdots\theta_{n_p}\right)\!.
\end{multline}
This observation leads us to the conclusion that representative elements in $H(D_P,\hA)$ depend only on even $\theta_{2k}$ and odd $\theta_{2k+1}$, for which we pick $\theta_0=\theta$ and $\theta_1$.

More formally, we introduce the family of homotopy operators
	\begin{equation}
	h_n=\frac{1}{2}\left(\sum_{r\geq 0}\frac{\dev}{\dev\theta_{n+2r+1}}-\sum_{r\leq0}\frac{\dev}{\dev\theta_{n-2r-1}}\right)\int d u_n.
	\end{equation}
	Such a definition is motivated by the observation that
	\begin{equation}\label{eq:poinclemma}
	\frac{1}{2}\left(\sum_{r\geq 0}\frac{\dev}{\dev\theta_{n+2r+1}}-\sum_{r\leq0}\frac{\dev}{\dev\theta_{n-2r-1}}\right)(\theta_{m+1}-\theta_{m-1})=\delta_{m,n}.
	\end{equation} 
	A direct computation, in which \eqref{eq:poinclemma} plays the central role, shows that
	\begin{equation}\label{eq:poinclemma2}
	h_nD_P+D_Ph_n=1-\pi\vert_{u_n=\emptyset}\pi\vert_{\theta_{n+1}=\theta_{n-1}},
	\end{equation}
	where by $\pi\vert_{u_n=\emptyset}$ we denote the projection which sets to 0 any element of $\hA$ with a dependancy on $u_n$ and by $\pi\vert_{\theta_{n+1}=\theta_{n-1}}$ the projection which replaces any occurrence of $\theta_{n+1}$ with $\theta_{n-1}$ (we can use the projection the other way round, replacing $\theta_{n-1}$ with $\theta_{n+1}$ if $n<0$). By \eqref{eq:poinclemma2}, for any element $f\in \ker D_P$ we have $f=\tilde{f}\vert_{\theta_{n+1}=\theta_{n-1}}+D_Pg$, where $\tilde{f}$ is the element $f$ in $\hA$ after removing all the dependancy on $u_n$ and $g\in\hA$. We can therefore pick as representatives for the cohomology classes in $H(D_P,\hA)$ the polynomials in the variables $\theta$ and $\theta_1$ alone. The property $\theta_n^2=0$, however, restricts the ring of such polynomials to a simple four dimensional vector space.
\end{proof}
\subsection{The cohomology $H(d_P,\hF)$}
The main result of this paper is a Theorem, alread stated in the Introduction, completely determining the Poisson cohomology of the Hamiltonian operator $K_0(\S)=\S-\S^{-1}$, namely $H(d_P,\hF)$ where with $d_P$ we denote the differential defined by the bivector $P$ as in \eqref{eq:normalform}. We have
{
	\renewcommand{\thetheorem}{\ref{main}}
	\begin{theorem}
		The Poisson cohomology defined by $P$ is finite dimensional. We have
		$$H^p(d_P,\hF)=0\qquad\qquad \forall\, p>1.$$
		Moreover,
		\begin{align}
		H^0(d_P,\hF)&=\left\{\int \alpha+\int\beta\, u\; \Big|\; (\alpha,\beta)\in\C^2\right\}\\
		H^1(d_P,\hF)&=\left\{\int\gamma \,\theta\;\Big|\;\gamma\in\C\right\}
		\end{align}
		The elements of the first cohomology group correspond to the evolutionary vector fields with characteristic $\gamma\in\R$.
	\end{theorem}
	\addtocounter{theorem}{-1}
}

Lemma \ref{poincare} states that $H^p(\hA)=0$ for $p>2$. This fact alone, given the long exact sequence \eqref{eq:longexact}, allows us to set the first estimation for $H(\hF)$, namely
$$
H^p(d_P,\hF)=0\qquad\qquad\forall\, p>2.
$$ 
In particular, $H^3(\hF)=0$, which means that all the infinitesimal deformations (both trivial and nontrivial) can be extended without obstructions. $H^0$, $H^1$, and $H^2$ require a more careful examination which is carried out in the following three lemmas. Together with $H^p(d_P,\hF)=0$ for $p>2$ they conclude the proofs of Theorem \ref{main}.
\begin{lemma}
	$H^0(d_P,\hF)=\langle\int1,\int u\rangle$
\end{lemma}
\begin{proof}
For $H^0$, the long exact sequence in cohomology reads
\begin{equation}
\xymatrix{
	H^0(\hA/\C)\ar[r]&H^0(\hA)\ar[r]^\int&H^0(\hF)\ar[r]^{\mathfrak{b}}&H^1(\hA/\C)\ar[r]^{\S-1}&H^1(\hA)\ar[r]&\dots.
}
\end{equation}
We first notice that for $p>0$ we have $\hA^p/\C\cong\hA^p$, and $H^0(\hA/\C)=0$. Therefore, $H^0(\hF)\cong H^0(\hA)\oplus N$, where $$
N\cong\im \mathfrak{b}\cong\ker (\S-1)\colon H^1(\hA)\to H^1(\hA)
$$ 
We can easily compute the induced map $\S-1$ on $H^1(\hA)$, since in it $\S\theta=\theta_1$ and $\S\theta_1=\theta$. We have
$$
(\S-1)(\alpha\theta+\beta\theta_1)=(\alpha-\beta)\theta_1-(\alpha-\beta)\theta,
$$
from which $N=\langle \theta+\theta_1\rangle$. We can therefore state that $H^0(d_P,\hF)\cong\C^2$. Moreover, we can identify representatives in the cohomology class: on one hand, we have the inclusion of $H^0(\hA)$ in $H^0(\hF)$ by the integral map, so $\int\alpha\in H^0(\hF)$ for all $\alpha\in\C$. On the other hand, we should identify the map $\mathfrak{b}$, such that $\mathfrak{b}^{-1}N\subset H^0(\hF)$. Let
$$\mathfrak{b}^{-1}:=\frac12\int u\left(\frac{\dev}{\dev\theta}+\frac{\dev}{\dev\theta_1}\right)
$$
for which we have $\mathfrak{b}^{-1}N=\int\beta u$, $\beta\in\C$. We immediately observe that $d_P\mathfrak{b}^{-1}N=\beta\int(\theta_1-\theta_{-1})=0$, so we have found the remaining $1$-dimensional subspace in $H^0(\hF)$.
\end{proof}
\begin{lemma}
	$H^1(d_P,\hF)=\langle \int\theta\rangle$
\end{lemma}
\begin{proof}
The relevant section of the long exact sequence is
\begin{equation}
\xymatrix{
	H^1(\hA)\ar[r]^{\S-1}&H^1(\hA)\ar[r]^\int&H^1(\hF)\ar[r]^{\mathfrak{b}}&H^2(\hA)\ar[r]^{\S-1}&H^2(\hA)\ar[r]&\dots.
}
\end{equation}
The crucial observation is that $\ker(\S-1)$ on $H^2(\hA)$ is $0$ and $(\S-1)H^2(\hA)\cong H^2(\hA)$, which means that the Bockstein homomorphism $\mathfrak{b}$ is the $0$ map. Hence $H^1(\hF)\cong\int H^1(\hA)$ and we immediately see that
$$
\int(\alpha\theta+\beta\theta_1)=(\alpha+\beta)\int \theta.
$$
\end{proof}
We have already discussed in Section \ref{sec:polyv} that the element $\int\theta\in\hF^1$ corresponds to the evolutionary vector field $\hat{X}(F)=\int\delta_uF$.
\begin{lemma}
	$H^2(d_P,\hF)=0$.
\end{lemma}
\begin{proof}
The vanishing of the higher groups of $H(\hA)$ implies that the long exact sequence terminates as
\begin{equation}
\xymatrix{
	0\ar[r]&H^2(\hA)\ar[r]^{\S-1}&H^2(\hA)\ar[r]^\int&H^2(\hF)\ar[r]&0,
}
\end{equation}
where the first arrow is determined by the previous remark that $\ker (\S-1)$ vanishes on $H^2(\hA)$. Finally, the short exact sequence implies that
$$
H^2(\hF)=\frac{H^2(\hA)}{(\S-1)H^2(\hA)}\cong\frac{H^2(\hA)}{H^2(\hA)}\cong0.
$$
\end{proof}
The vanishing of the second cohomology group implies that all the infinitesimal deformations $\tilde{P}$ of the bivector $P$ are trivial, namely that there exists an evolutionary vector field $X$ such that $\tilde{P}=d_PX$.

\section{Applications: compatible low order scalar Hamiltonian operators}\label{sec:applic}
In Theorem \ref{main}, the main result of our paper, we obtained the Poisson cohomology of the operator $K_0=\S-\S^{-1}$. We are going to apply it to compatible \emph{Hamiltonian} operators, that form a biHamiltonian pair with $K_0$.

In their recent paper \cite{dSKVW18}, De Sole, Kac, Valeri and Wakimoto present a classification of local difference Hamiltonian structures, in terms of multiplicative Poisson vertex algebras, complete up to order $(-5,5)$. In this Section we discuss, in light of our theorem, the compatible pairs in their classification, that we present in Appendix \ref{app:compPairs}. 

We denote the Hamiltonian operators corresponding to the bracket $\{u_\lambda u\}_{k,g}$ in \eqref{eq:Kgen} by $K_{k,g}$, the one corresponding to \eqref{eq:K2g} as $\tilde{K}_{2,g}$, the one given by \eqref{eq:K3g} as $\tilde{K}_{3,g}$. The operator given by \eqref{eq:Qg} will be denoted as $Q_g$ and the one given by \eqref{eq:QQg} by $\tilde{Q}_g$.

We have proved in Section 3 that the normal form for a Hamiltonian operator of order $(-1,1)$ is $K_0$. Indeed, $K_0$ can be obtained by the change of coordinates $$u\mapsto v=\int \frac{1}{g(u)},$$
as explicitly shown in \eqref{eq:chNF}.

The same change of coordinates can be applied to any of the bi-Hamiltonian pairs $(K_{1,g},K_g)$, $(K_{1,g},\tilde{K}_{2,g})$, $(K_{1,g},\tilde{K}_{3,g})$, $(\tilde{K}_{2,g},Q_g)$, and $(K_{1,g}+K_{2,g},\tilde{Q}_g)$ to get the pairs defined by the operators $K_0=K_{1,1}$, $K_1$, $\tilde{K}_{2,1}$, $\tilde{K}_{3,1}$, $Q_1$, and $\tilde{Q}_1$. The bracket corresponding to each Hamiltonian operator is obtained computing
$$
\{v_{\lambda}v\}=\left\{\left(\int\frac{1}{g(u)}\right){}_\lambda\left(\int\frac{1}{g(u)}\right)\right\}
$$
and expressing the result using the new coordinate $v$. The computation is performed exploited the so-called master formula \eqref{eq:master}; we demonstrate it for $\tilde{K}_{2,g}$ in Appendix \ref{app:chC}.

Theorem \ref{main} states that there exist evolutionary vector fields whose action maps each of the operators compatible with $K_0$ (namely, $K_1$, $\tilde{K}_{2,1}$, and $\tilde{K}_{3,1}$) to $K_0$.
\subsection{Constant compatible Hamiltonian operators}\label{ssec:const}

Let us consider the operator $K_1=\sum_{k=2}^5 c_k(\S^k-\S^{-k})$ with arbitrary constants $c_2,\ldots,c_5$. The compatibility holds true for any choice of the constants; hence each of the homogeneous summands is compatible with $K_0$. The homogeneous term of order $(-k,k)$ is $K_{k,1}$.

We denote $P_0$ and $P_{1,k}$, respectively, the Poisson bivectors corresponding to the operators $K_0$ and $K_{k,1}$. From the vanishing of the Poisson cohomology of $P_0$, it follows that there exists a vector field $X^{(k)}$ such that $P_{1,k}=[P_0,X^{(k)}]$. Such vector field generates the infinitesimal change of coordinates under which $K_0$ becomes $K_{k,1}$. In terms of $\lambda$ brackets for multiplicative Poisson vertex algebras, we look for a function $f^{(k)}(u,u_1,u_2,\ldots)$ such that
\begin{equation}\label{eq:adjact-bracket-1}
\left\{\left(u+\epsilon f^{(k)}\right){}_\lambda \left(u+\epsilon f^{(k)}\right)\right\}_{0}=\{u_\lambda u\}_{0}+\epsilon\{u_\lambda u\}_{k,1}+O(\epsilon^2),
\end{equation}
where we denote $\{u_{\lambda}u\}_0$ the $\lambda$ bracket called $\{u_{\lambda}u\}_{1,1}$ in Appendix \ref{app:compPairs}. The term of order $\epsilon$ in \eqref{eq:adjact-bracket-1} is
$$
\{{f^{(k)}}{}_{\lambda}\, u\}_0+\{u\,{}_{\lambda}\,{f^{(k)}}\}_0=\{u_{\lambda}u\}_{k,1}
$$
from which we obtain, using the master formula \eqref{eq:master}, the set of equations
\begin{equation}
\frac{\dev f^{(k)}}{\dev u_{n-1}}-\frac{\dev f^{(k)}}{\dev u_{n+1}}=\delta_{k,n}.
\end{equation}
Its solution is
\begin{equation}
f^{(2k)}=\sum_{n=0}^{k-1}u_{2n+1},\qquad
f^{(2k+1)}=\sum_{n=0}^{k-1}u_{2n};
\end{equation}
The vector field $X^{(k)}$ producing the infinitesimal change of coordinates $u\mapsto u+\epsilon f^{(k)}$ is $-\int f^{(k)}\theta$.

From the Jacobi identity of the Schouten bracket, $[P_{k'},X^{(k)}]$ is compatible with $P_{0}$ for any $k,k'$. We have
\begin{multline}
0=[[P_{k'},X^{(k)}],P_0]+[[P_0,P_{k'},],X^{(k)}]+[[X^{(k)},P_0],P_{k'}]=\\
[[P_{k'},X^{(k)}],P_0]+[P_k,P_{k'}]=[[P_{k'},X^{(k)}],P_0].
\end{multline}
This means that vector fields $\tilde{X}^{(k)}$ can always be found in such a way that there exists a change of coordinates $\phi=e^{\mathrm{ad}_X}$, $X=\sum \tilde{X}^{(k)}$, connecting a Hamiltonian operator $K_{1}$ with the normal form $K_0$. Note that $\tilde{X}^{(k)}$ will be some multiple of $X^{(k)}$, with coefficient depending on $c_k$.
\subsection{$N=2$ compatible operators and the bi-Hamiltonian structure of the Volterra chain}\label{ssec:biHamVolt}
We consider the well-known Volterra chain, which is the evolutionary differential-difference equation
\begin{equation}
\label{eq:volterra1}
\frac{\dev u}{\dev t}=u(u_{1}-u_{-1}).
\end{equation}
It is well known (see, for instance, \cite{CMW18}) that the Volterra chain is an integrable equation admitting a biHamiltonian formulation. Indeed, one can write \eqref{eq:volterra1} with respect to two compatible Hamiltonian operators
\begin{align}
K_{1,u}&=uu_1\S-u_{-1}u\S^{-1},\\
\tilde{K}_{2,u}&=uu_1u_2\S^2-uu_{-1}u_{-2}\S^{-2}+uu_1\left(u+u_1\right)\S-uu_{-1}\left(u+u_{-1}\right)\S^{-1}
\end{align}
and the two Hamiltonian functionals
\begin{align}
H_1&=\int u&H_0&=\frac12\int\log u
\end{align}
according to $\dev_tu=K_{1,u}\delta H_1=\tilde{K}_{2,u} \delta H_0$.

The operator $K_{1,u}$ is of order $(-1,1)$, so there exist a coordinate change of the form $v=\log u$ for which it takes the constant form $K_0=\S-\S^{-1}$, corresponding to the multiplicative $\lambda$ bracket $\{v_\lambda v\}_1=\lambda-\lambda^{-1}$.

The same change of coordinates for $\tilde{K}_{2,u}$ gives
\begin{multline}
\{v_\lambda v\}_2=u_1\lambda^2+(u+u_1)\lambda-(u+u_{-1})\lambda^{-1}-u_{-1}\lambda^{-2}\\
=e^{v_1}\lambda^2+\left(e^v+e^{v_1}\right)\lambda-\left(e^v+e^{v_{-1}}\right)\lambda^{-1}-e^{v_{-1}}\lambda^{-2},
\end{multline}
namely the bracket for $\tilde{K}_{2,1}$. The compatibility of $K_0$ and $\tilde{K}_{2,1}$ can be explicitly verified and it is indeed case (2) of the classification of the compatible pairs for $g(u)=1$. The Poisson bivector corresponding to $\tilde{K}_{2,1}$ is
$$
P_2=\frac12\int\left( e^{v_1}\theta\theta_2+\left(e^v+e^{v_1}\right)\theta\theta_1\right)
$$

The compatibility of the two Hamiltonian operators is equivalent to $d_{P_0}P_2=0$. From Theorem \ref{main} we conclude that there must exist an evolutionary vector field $X$ such that $P_2=[P_0,X]$, or, equivalently, that exist $f(v,v_1,\ldots)$ as in \eqref{eq:adjact-bracket-1} such that
\begin{equation}\label{eq:miura}
\{f_\lambda v\}_0+\{v_\lambda f\}_0=\{v_\lambda v\}_2
\end{equation}
The evolutionary vector field $X=-\int f\theta$ is 
$$
X=-\int\left(e^v+e^{v_1}\right)\theta.
$$
This can easily be verified computing $[P_0,X]$:
\begin{multline}
[P_0,X]=\int\left(\frac{\delta P_0}{\delta\theta}\frac{\delta X}{\delta v}\right)=\frac12\int\left(\theta_{-1}-\theta_{1}\right)e^v\theta+\frac12\int\left(\theta_{-1}-\theta_{1}\right)e^v\theta_{-1}=\\
\frac12\int\left(e^v\theta\theta_1+e^v\theta_{-1}\theta+e^v\theta_{-1}\theta_1\right)=\frac12\int\left(\left(e^v+e^{v_1}\right)\theta\theta_1+e^{v_1}\theta\theta_2\right).
\end{multline}

\subsection{$N=3$ compatible operators}
The $(-3,3)$ order Hamiltonian operator $\tilde{K}_{3,1}$
\begin{multline}
\tilde{K}_{3,1}=e^{u_1+\I u_2}\S^3+\I\left(e^{u+\I u_1}-e^{u_1+\I u_2}\right)\S^2+e^{u+\I u_1}\S\\
-e^{u_{-1}+\I u}\S^{-1}-\I\left(e^{u_{-2}+\I u_{-1}}-e^{u_{-1}+\I u}\right)\S^{-2}-e^{u_{-2}+\I u_{-1}}\S^{-3}
\end{multline}
is compatible with $K_0$, according to the classification of compatible Hamiltonian pairs (case (3)). The vector field $X=-\int f\theta$ mapping the bivector $P_0$ into the one associated with $\tilde{K}_{3,1}$ is the solution of
\begin{align*}
\frac{\dev f}{\dev u}+\S\left(\frac{\dev f}{\dev u}\right)&=e^{u+\I u_1}+e^{u_1+\I u_2},&\quad\frac{\dev f}{\dev u_1}&=\I\left(e^{u+\I u_1}-e^{u_1+\I u_2}\right),\\
\frac{\dev f}{\dev u_2}&=e^{u_1+\I u_2},&\frac{\dev f}{\dev u_n}&=0\qquad n>2.
\end{align*}
The system is easy to solve:
\begin{equation}
f=e^{u+\I u_1}-\I e^{u_1+\I u_2}.
\end{equation}
The integrable hierarchy generated by the biHamiltonian pair $(K_0,\tilde{K}_{3,1})$ is equivalent to a stretched version of the Volterra chain. We pick $H_{-1}=\int u$ a Casimir of $K_0$. Introducing the new variable $w=u+\I u_1$ we have
\begin{equation}
\frac{\dev u}{\dev t}=\tilde{K}_{3,1}\delta H_{-1}=(1-\I)e^{w_1}+(1+\I)e^w-(1-\I)e^{w_{-1}}-(1+\I)e^{w_{-2}},
\end{equation}
namely
\begin{equation}
\frac{\dev w}{\dev t}=(1+\I)\left(e^{w_2}-e^{w_{-2}}\right).
\end{equation}
The change of coordinates $e^w\mapsto v$, $(1+\I)t\mapsto \tau$, finally, gives the equation
\begin{equation}\label{eq:volterra2}
\frac{\dev v}{\dev \tau}=v(v_2-v_{-2}),
\end{equation}
which is the Volterra chain equation \eqref{eq:volterra1} with a rescaled independent variable.
\begin{remark}
	The third order Hamiltonian operator of complementary type is not in the Volterra hierarchy generated by the biHamiltonian pair of the first and second order ones. We have observed that $\tilde{K}_{2,u}$ is the second Hamiltonian structure of the Volterra chain, while $K_{1,u}$ is the first one. It is known that the compatibility of $K_{1,u}$ and $\tilde{K}_{2,u}$ implies the existence of the recursion operator $R=\tilde{K}_{2,u}K_{1,u}^{-1}$; such operator is nonlocal, i.e. it is not a polynomial in $\S$ and $\S^{-1}$, and hence in principle the higher Hamiltonian structures in the hierarchy $K_m=R^{m-1}K_1$ are not local. The explicit form of the recursion operator is
	$$
	R=u\S+u+u_1+u\S^{-1}+u(u_1-u_{-1})(\S-1)^{-1}\frac1u.
	$$
	On the other hand, in the same coordinate system we have
	\begin{multline}
	\tilde{K}_{3,u}=uu_1u_2^\I u_3\S^3-u_{-3}u_{-2}u_{-1}^\I u\S^{-3}\\
	+\I\left(u^2u_1^\I u_2-uu_1u_2^{1+\I}\right)\S^2-\I\left(u_{-2}^2u_{-1}^\I u-u_{-2}u_{-1}u^{1+\I}\right)\S^{-2}\\
	+u^2u_1^{1+\I}\S-u_{-1}^2u^{1+\I}\S^{-1},
	\end{multline}
	which is obviously different from $RK_2$.
\end{remark}
\subsection{Compatible $N=2$ and $N=4$ order operators}
The $(-2,2)$ order Hamiltonian operator $\tilde{K}_{2,1}$ is compatible with the $(-4,4)$ order operator $Q_1$
\begin{multline}
	Q_1=e^{u_1-u_2+u_3}\S^4-\S^{-4}e^{u_1-u_2+u_3}+\left(e^{u-u_1+u_2}+e^{u_1-u_2+u_3}\right)\S^3+\\
-\S^{-3}\left(e^{u-u_1+u_2}+e^{u_1-u_2+u_3}\right)+e^{u-u_1+u_2}\S^2-\S^{-2}e^{u-u_1+u_2}.
\end{multline}
We have already shown that $\tilde{K}_{2,1}$ is compatible with $K_0$, so there exist a change of coordinates $u=F(v)$ such that the operator $\tilde{K}_{2,g}$ in the new coordinates is of the form $K_0$. Performing the same change of coordinates on $Q_1$ we then find a compatible Hamiltonian pair.

Given a change of coordinates $v=F(u,u_1,\ldots)$, the transformation law for difference operators is
\begin{equation}\label{eq:transfK}
F_* K\vert_{u} F_*^\dagger = K'\vert_{F(u)},
\end{equation}
where $F_*$ is the Fr\'echet derivative of the change of coordinates and $F_*^\dagger$ is its formal adjoint
$$
F_*=\sum_{n\in \Z}\frac{\dev F}{\dev u_n}\S^n\qquad\qquad F_*^\dagger=\sum_{n\in\Z}\S^{-n}\frac{\dev F}{\dev u_n}
$$
We observe that, for $K=K_0$ and the change of coordinates 
$$v=F(u,u_1)=-\log(u)-\log(u_1),$$
the Fr\'echet derivative and its adjoint are
$$
F_*=-\frac{1}{u}-\frac1u_1\S\qquad\qquad F_*^\dagger=-\frac1u-\frac1u\S^{-1},
$$
which gives the transformed operator
\begin{equation}
K'=\frac{1}{u_1u_2}\S^2+\left(\frac{1}{uu_1}+\frac{1}{u_1u_2}\right)\S
-\S^{-1}\left(\frac{1}{uu_1}+\frac{1}{u_1u_2}\right)-\S^{-2}\frac{1}{u_1u_2}.
\end{equation}
It is now easy to see that, in the new coordinates, we can write $K'$ as
$$
K'\vert_v=e^{v_1}\S^2+\left(e^v+e^{v_1}\right)\S-\S^{-1}\left(e^v+e^{v_1}\right)-\S^{-2}e^{v_1}=\tilde{K}_{2,1}\vert_v.
$$

The same change of coordinates $u=F^{-1}(v)$ transforming $\tilde{K}_{2,1}$ into $K_0$ can be applied to $Q_1$ to find the corresponding compatible operator. However, it is easier not to invert $F$ and to look instead for the preimage of the transformed operator. We have found that the change of coordinates $v=F(u)$ raises the order of the difference operator from $(-1,1)$ to $(-2,2)$. For this reason, the ansatz for the unknown operator of which $Q_1$ is the transformed under $F$ is the $(-3,3)$ order operator
$$
Q'= A\S^3+B\S^2+C\S-\S^{-1}C-\S^{-2}B-\S^{-3}A,
$$ 
with $A$, $B$ and $C$ in $\A$. We equate $F_* Q' F_*^\dagger$ to $Q_1\vert_{F(u)}$ to find the coefficients $A$, $B$ and $C$ of $Q'$.

We have
\begin{align}
Q_1\vert_v&=e^{v_1-v_2+v_3}\S^4+\left(e^{v-v_1+v_2}+e^{v_1-v_2+v_3}\right)\S^3+e^{v-v_1+v_2}\S^2\\
&\quad-\S^{-2}e^{v-v_1+v_2}-\S^{-3}\left(e^{v-v_1+v_2}+e^{v_1-v_2+v_3}\right)-\S^{-4}e^{v_1-v_2+v_3},\\
Q_1\vert_{F(u)}&=\frac{1}{u_1u_4}\S^4+\left(\frac{1}{uu_3}+\frac{1}{u_1u_4}\right)\S^3+\frac{1}{uu_3}\S^2\\
&\quad-\S^{-2}\frac{1}{uu_3}-\S^{-3}\left(\frac{1}{uu_3}+\frac{1}{u_1u_4}\right)-\S^{-4}\frac{1}{u_1u_4}.
\end{align}
Comparing the coefficients of $\S^4$, $\S^3$, $\S^2$ and $\S$ in the two sides of \eqref{eq:transfK} we find
$$
A=1\qquad\qquad B=0\qquad\qquad C=0
$$
which gives us
$$
Q'=\S^3-\S^{-3}=K_{3,1}.
$$
\subsection{Compatible $N=2$ and $N=5$ order operators}
The non-homogeneous $(-2,2)$ order operator $K_{1,1}+K_{2,1}=\S^2+\S^1-\S^{-1}-\S^{-2}$ is compatible with the $(-5,5)$ order operator $\tilde{Q}_1$
\begin{align}
\tilde{Q}_1&=e^{\varepsilon u_2+u_3}\S^5-\left(\varepsilon e^{\varepsilon u_1+u_2}+\varepsilon^{-1}e^{\varepsilon u_2+u_3}\right)\S^4+\\ \notag 
&\quad+\left(\varepsilon^{-1} e^{\varepsilon u+u_1}+e^{\varepsilon u_1+u_2}+\varepsilon e^{\varepsilon u_2+u_3}\right)\S^3-\left(\varepsilon e^{\varepsilon u+u_1}+\varepsilon^{-1}e^{\varepsilon u_1+u_2}\right)\S^2\\
&\quad+e^{\varepsilon u+u_1}\S-\S^{-1}e^{\varepsilon u+u_1}+\S^{-2}\left(\varepsilon e^{\varepsilon u+u_1}+\varepsilon^{-1}e^{\varepsilon u_1+u_2}\right)\\
&\quad-\S^{-3}\left(\varepsilon^{-1} e^{\varepsilon u+u_1}+e^{\varepsilon u_1+u_2}+\varepsilon e^{\varepsilon u_2+u_3}\right)\\
&\quad+\S^{-4}\left(\varepsilon e^{\varepsilon u_1+u_2}+\varepsilon^{-1}e^{\varepsilon u_2+u_3}\right)-\S^{-5}e^{\varepsilon u_2+u_3},
\end{align}
where we denote by $\varepsilon$ a primitive $3$-rd root of unity.

Choosing $H_{-1}=\int u$ the Casimir of the first bracket, and introducing the new variable $v=\varepsilon u+u_1$ we have
\begin{align}
\frac{\dev u}{\dev t}=\tilde{Q}_1\delta H_{-1}&=\left(1-\varepsilon^{-1}+\varepsilon\right)e^{v_2}+2e^{v_1}+\left(1+\varepsilon^{-1}-\varepsilon\right)e^{v}\\
&\quad-\left(1-\varepsilon^{-1}+\varepsilon\right)e^{v_{-1}}-2e^{v_{-2}}-\left(1+\varepsilon^{-1}-\varepsilon\right)e^{v_{-3}},
\end{align}
namely
\begin{equation}
\frac{\dev v}{\dev t}=-2\varepsilon^{-1}\left(e^{v_3}-e^{v_{-3}}\right).
\end{equation}
Finally, the change of coordinates $e^v\mapsto w$, $-2\varepsilon t\mapsto\tau$ gives us
\begin{equation}
\frac{\dev w}{\dev \tau}=\varepsilon w (w_3-w_{-3}),
\end{equation}
which is yet another Volterra chain equation with rescaled variables, see for instance \eqref{eq:volterra2}.

The same change of coordinates $v(u)=\varepsilon u+u_1$ brings, according to the transformation law \eqref{eq:transfK}, the operator $K_{1,1}+K_{2,2}$ to
\begin{equation}
K_0^{(\varepsilon,3)}:=\varepsilon(\S^3-\S^{-3})=\varepsilon K_{3,1}
\end{equation}
and the Hamiltonian operator $\tilde{Q}_1$ to
\begin{equation}
\tilde{K}_{2,1}^{(\varepsilon,3)}=\varepsilon \tilde{K}_{2,1}^{(3)}:=\varepsilon\left(e^{v_3}\S^6+\left(e^{v_3}+e^v\right)\S^3-\S^{-3}\left(e^{v_3}+e^v\right)-\S^{-6}e^{v_3}\right),
\end{equation}
which are a rescaled and $3$-stretched (namely, the independent variable is rescaled according to $n\mapsto 3n$) version of the biHamiltonian pair of Volterra equation when $g(u)=1$, previously discussed in Section \ref{ssec:biHamVolt}.
\begin{remark}
The rescaling of the independent variable that from a Hamiltonian operator $K$ produces a $k$-\emph{stretched} operator $K^{(k)}$ by $n \mapsto k\,n$ cannot, in general, be obtained by a change of dependent variables as the ones which existence is guaranteed by the vanishing of the second and third Poisson cohomology groups. It is possible for the case we are considering, because $K^{(3)}_{1,1}$ is compatible with $K_0$ (we have dropped the further rescaling by $\varepsilon$). From Theorem \ref{main} it follows that there exists an evolutionary vector field such that $K^{(3)}_{1,1}=[X,K_0]$, and an easy computation shows that such vector field is
\begin{equation}
X=\int \left(v_2-\frac v2\right)\theta.
\end{equation}
However, $\tilde{K}_{2,1}^{(3)}$ is neither compatible with $K_0$ nor with $\tilde{K}_{2,1}$, hence Theorem \ref{main} does not provide any further insight into it.
\end{remark}
\section{The Poisson cohomology for stretched Hamiltonian operators}\label{sec:stretch}
Given a Hamiltonian operator of order $(-N,N)$ of the form
\begin{equation}
K(\S)=\sum_{n=-N}^N a^{(n)}(u,u_1,u_{-1},u_2,u_{-2},\ldots,u_m,u_{-m},\ldots) \S^n,
\end{equation}
we define its \emph{k-stretched} version by rescaling the underlying lattice variable by a factor $k$, namely
\begin{equation}
K^{(k)}(\S):=\sum_{n=-N}^N a^{(n)}(u,u_k,u_{-k},u_{2k},u_{-2k},\ldots,u_{mk},u_{-mk},\ldots)\S^{nk}.
\end{equation}
The computation of the Poisson cohomology for the $k$-stretched $K_0^{(k)}=\S^{k}-\S^{-k}$ follows the lines of the proof of Theorem \ref{main}, but we obtain an essentially different result. The Poisson bivector corresponding to the $K_0^{(k)}$ is $P=\int\theta\theta_k$, definining on $\hA$ the differential
$$
D_P=\sum_n\left(\theta_{n+k}-\theta_{n-k}\right)\frac{\dev}{\dev u_n}.
$$
The cohomology $H(D_P,\hA)$ is, hence, the polynomial ring generated by the $2k$ variables $\{\theta,\theta_1,\ldots,\theta_{2k-1}\}$. This implies $\dim H(D_P,\hA)=2^{2k}$ and $\dim H^p(D_P,\hA)=\binom{2k}{p}$. Using the same long exact sequence argument we exploited in the proof of the main theorem, we have the first general result on the Poisson cohomology for $P$, namely
\begin{theorem}\label{thm:stretch}
	$H^p(\ud_P,\hF)=0$ for $p>2k$.
\end{theorem}
As in the $(-1,1)$ case, the lower cohomology groups should be computed explicitly. We observe, however, that already in the $2$-stretched case $H^2(\ud_P,\hF)\neq 0$.

We have indeed $P=\int\theta\theta_2$ and
$$
H^2(D_P,\hA)=\left\langle\theta\theta_1,\theta\theta_2,\theta\theta_3,\theta_1\theta_2,\theta_1\theta_3,\theta_2\theta_3\right\rangle.
$$
The relevant part of the long exact sequence in cohomology is the following
\begin{equation}
\xymatrix{
	\dots\ar[r]&H^2(\hA)\ar[r]^{\S-1}&H^2(\hA)\ar[r]^\int&H^2(\hF)\ar[r]^{\mathfrak{b}}&H^3(\hA)\ar[r]^{\S-1}&\dots.
}
\end{equation}
from which we have that 
$$
H^2(\hF)=\int H^2(\hA)\oplus\mathfrak{b}^{-1}\left(\left.\ker(\S-1)\right|_{H^3(\hA)}\right). 
$$
The integral map on $H^2(\hA)$ gives $\theta\theta_1\sim\theta_1\theta_2\sim\theta_2\theta_3\sim -\theta\theta_3$ and $\theta\theta_2\sim\theta_1\theta_3\sim-\theta\theta_2\sim0$, corresponding to the bivector $\int\theta\theta_1$. On the other hand, $H^3(\hA)$ is the 4-dimensional vector space $\langle\theta\theta_1\theta_2,\theta\theta_1\theta_3,\theta\theta_2\theta_3,\theta_1\theta_2\theta_3\rangle$. The kernel of $(\S-1)$ in the cohomology is 1-dimensional, explicitly $\langle\theta\theta_1\theta_2+\theta\theta_1\theta_3+\theta\theta_2\theta_3+\theta_1\theta_2\theta_3\rangle$. These observations are enough to state the nontriviality of the second Poisson cohomology group for the $2$-stretched operator, and in particular that it is $2$-dimensional. The constant bivector in $H^2(\hF)$ corresponds to the operator $K_0$.

In Section \ref{sec:Poisson} we have explained that 2-cochains correspond to nontrivial infinitesimal deformations of the Poisson bivector. Our result means that there does not exists an evolutionary vector field $\hat{X}$ such that $[\hat{X},P]=\int\theta\theta_1$. However, we can find a \emph{formal} vector field $\hat{X}_f$ solution of
\begin{equation}\label{eq:formvec}
[\hat{X}_f,P]=\int\frac{\delta \hat{X}_f}{\delta u}\frac{\delta P}{\delta \theta}=\int\theta\theta_1.
\end{equation}
Indeed, \eqref{eq:formvec} is equivalent to the following set of equations for $\hat{X}_f=\int X_f\theta$
\begin{equation}
\S^m\left(\frac{\dev X_f}{\dev u_{-1-m}}-\frac{\dev X_f}{\dev u_{3-m}}\right)=\delta_{m,0}.
\end{equation}
A solution of this system is $X_f=\sum_{m\geq 0}u_{-4m-1}$. The formal vector turns out to be an infinite sum of vector fields of characteristic $u_{-4m-1}$. This result does not contradict our statement about Poisson cohomology, since by definition the characteristic of a vector field is an element of the algebra of difference functions, and suche elements depend only on a finite number of varaibles $u_n$.
\begin{remark}
	The structure of the Poisson cohomology for the stretched operator has a striking resemblance to the Poisson cohomology of constant higher order scalar \emph{differential} Hamiltonian operators. In \cite[Remark 15]{CCS15} the Poisson cohomology for the first order scalar differential operator is obtained with a similar approach to the one adopted in this paper. The procedure can be repeated for higher order differential operators of the form $Q_k(\dev)=\dev^{2k+1}$. In this case we obtain $H^p(Q_k)=0$ for $p>2k+1$, too. Similarly, $H^2(Q_k)\neq 0$ for $k>0$, and in particular the first order scalar differential operator $\dev$ is a cocycle of $Q_1$. However, a major difference between the difference and the differential case is the lack of the formal vector field mapping $Q_1$ into $Q_0$.
\end{remark}
\section{Final remarks}
The results obtained in this paper constitute an analogue of Getzler's theorem for the Poisson cohomology of Hamiltonian operators, in the restricted contest of the scalar case, in the difference setting.

The vanishing of the second and third cohomology group for $K_0=\S-\S^{-1}$ means that there exist a Miura-type change of coordinates $u\mapsto f(u,u_{-1},u_{1},\ldots)$ under which a Hamiltonian operator compatible with $K_0$ is of the form $K_0$. Moreover, we have explicitly shown in Section 3 that any scalar difference Hamiltonian of order $(-1,1)$ can be brought to the form $K_0$ with a change of coordinates, hence implying that the aforementioned result holds for any Hamiltonian operator of this order.
However, the change of coordinates can be formal, as discussed in Section 4.2; this is already well known in the differential case, because the Poisson cohomology guarantees the existence of vector fields, for which the change of coordinates is the action of their exponential map. 
%

De Sole, Kac, Valeri and Wakimoto \cite{dSKVW18} have obtained the general form for scalar Hamiltonian operators from order $(-1,1)$ up to $(-5,5)$. Using the theory of multiplicative Poisson vertex algebras we see the dependancy on an arbitrary function can be removed by a change of coordinates \eqref{eq:chNF} -- one can verify on a case-by-case basis that the same result holds for all the operators in the classification --, while the theorem we proved implies that there exists an (operator-dependent) system of coordinates where any operator compatible with $K_0$ is of the form $K_0$. This does not affect the Hamiltonian operators of their classification which are not compatible with the $(-1,1)$ order one (and are of order $(-4,4)$ and $(-5,5)$).

Another minor contribution present in this paper is the observation that the integrable hierarchy defined by the biHamiltonian pair constituted by a $(-1,1)$ Hamiltonian operator and a compatible $(-3,3)$ order one is, despite the different order of the Hamiltonian operators, equivalent to the Volterra hierarchy.

All the results obtained in this paper apply to \emph{local} Hamiltonian operators, namely to operators which are polynomials in $\S$ and $\S^{-1}$. A theory of rational Hamiltonian operator has been recently developed \cite{CMW18-2}; the Poisson cohomology of such a larger class of structures is the natural further topic in the direction of their classification.

\appendix
\section{Multiplicative Poisson vertex algebras and Hamiltonian structures}\label{app:mPVA}
In Section \ref{sec:formalism} and Section \ref{sec:applic} we have adopted the formalism of multiplicative Poisson vertex algebras (multiplicative PVAs), introduced by De Sole, Kac, Valeri and Wakimoto \cite{dSKVW18}, as a computational tool for operating on difference operators.

In this appendix we present the main results we need and that we have referenced in the paper.
\subsection{The definition}\label{app:mPVA-1}
In this paragraph we present the definition of multiplicative PVAs and the so-called \emph{master formula}, the main computational tool they provide. The equivalence between Hamiltonian operators, Poisson bivectors and multiplicative PVAs has been discussed in Proposition \ref{thm:OpBiv}.

\begin{definition} A multiplicative PVA is an algebra of difference functions $(\A,\S)$ endowed with a $\R$-bilinear operation
\begin{equation}
\begin{split}
\{\phantom{f}_\lambda\phantom{g}\}\colon \A\times\A&\to\A[\lambda,\lambda^{-1}]\\
\{f_\lambda g\}&:=\sum_{s=-N}^Nc(f,g)_{(s)}\lambda^s
\end{split}
\end{equation}
called the \emph{$\lambda$ bracket}, satisfying the properties
\begin{enumerate}
	\item (sesquilinearity) $\{\S f_\lambda g\}=\lambda^{-1}\{f_\lambda g\}$, $\{f_{\lambda}\S g\}=(\lambda\S)\{f_{\lambda}g\}$;
	\item (left Leibniz rule) $\{f_\lambda gh\}=\{f_\lambda g\}h+g\{f_\lambda h\}$;
	\item (right Leibniz rule) $\{fg_\lambda h\}=\{f_{\lambda\S} h\}g+\{g_{\lambda\S} h\}f$ -- this should be interpreted as $\{f_{\lambda\S}h\}g=\sum c(f,h)_{(s)}\S^sg\lambda^s$;
	\item (skewsymmetry) $\{g_\lambda f\}=-{}_\to\{f_{(\lambda\S)^{-1}}g\}$, where the right hand side should be read as $-\sum(\lambda \S)^s c(f,g)_{(s)}$;
	\item (PVA-Jacobi identity) $\{f_\lambda\{g_\mu h\}\}-\{g_\lambda\{f_\lambda h\}\}=\{\{f_\lambda g\}_{\lambda\mu}h\}$.
\end{enumerate}
\end{definition}
A Hamiltonian ($\ell\times\ell$ matrix of) difference operator $K^{ij}$ defines a multiplicative PVA by letting
\begin{equation}
\{u^i{}_\lambda u^j\}:=K^{ji}\big|_{\S\to\lambda}=K^{ji}(\lambda)
\end{equation} 
and then extending the bracket from the generators of $\A$ to the full algebra according to the properties (1)--(3).

The expression for the bracket on the full algebra $\A$ is called the \emph{master formula} and it has the form
\begin{equation}\label{eq:master}
\{f_{\lambda}g\}=\sum_{i,j=1}^\ell\sum_{n,m\in\Z}\frac{\dev g}{\dev u^j_m}(\lambda\S)^m\{u^i{}_{\lambda\S}u^j\}(\lambda\S)^{-n}\frac{\dev f}{\dev u^i_n}.
\end{equation}

In particular, the condition of being an Hamiltonian operator for $K$ is equivalent the PVA-Jacobi identity for any triple of generators $(u^i,u^j,u^k)$ \cite{dSKVW18}. Explicitly, we have
\begin{multline}
\label{eq:PVAJacobiHam}
\sum_{l=1}^\ell\sum_{n\in\Z}\left(\frac{\dev K^{kj}(\mu)}{\dev u^l_n}(\lambda \S)^nK^{li}(\lambda)-\frac{\dev K^{ki}(\lambda)}{\dev u^l_n}(\mu \S)^nK^{lj}(\mu)\right)=\\
\sum_{l=1}^\ell\sum_{n\in\Z} K^{kl}(\lambda\mu\S)(\lambda\mu\S)^{-n}\frac{\dev K^{ji}(\lambda)}{\dev u^l_n}.
\end{multline}

\subsection{Equivalence between Schouten and PVA-Jacobi identities in the scalar case}\label{app:PVAJacobi}
In the scalar case, the PVA-Jacobi identity for the multiplicative $\lambda$ bracket defined by the difference operator $K=\sum_{s}K^{(s)}\S^s$ is
\begin{multline}
\sum_{n,m,q\in\Z}\left(\frac{\dev K^{(m)}}{\dev u_n}\left(\S^nK^{(q)}\right)\lambda^{n+q}\mu^m-\frac{\dev K^{(m)}}{\dev u_n}\left(\S^nK^{(q)}\right)\mu^{n+q}\lambda^m-\right.\\
\left. K^{(m)}\left(\S^{m-n}\frac{\dev K^{(q)}}{\dev u_n}\right)\lambda^{m+q-n}\mu^{m-n}\right)=0.
\end{multline}
The skewsymmetry property \eqref{eq:skewBiv} imposes a special form for a scalar bracket, namely $K^{(-s)}=-\S^{-s}K^{(s)}$, or
\begin{equation}\label{eq:skewbracket-det}
\{u_{\lambda}u\}=\sum_{s>0} K^{(s)}\lambda^s-\S^{-s}K^{(s)}\lambda^{-s}.
\end{equation}
which gives the explicit form for the PVA-Jacobi identity
\begin{multline}\label{eq:pvapf}
-\frac{\dev K^{(m)}}{\dev u_n}\S^n K^{(q)} \,\boldsymbol{\lambda\mu}(m,n+q)+\frac{\dev K^{(m)}}{\dev u_n}\S^{n-q}K^{(q)}\,\boldsymbol{\lambda\mu}(m,n-q)\;-\\
(\S^nK^{(q)})\left(\S^{-m}\frac{\dev K^{(m)}}{\dev u_{n+m}}\right)\,\boldsymbol{\lambda\mu}(n+q,-m)\;+\\
(\S^{n-q}K^{(q)})\left(\S^{-m}\frac{\dev K^{(m)}}{\dev u_{n+m}}\right)\,\boldsymbol{\lambda\mu}(n-q,-m)-K^{(q)}\left(\S^{q-n}\frac{\dev K^{(m)}}{\dev u_n}\right)\,\boldsymbol{\lambda\mu}(q-n+m,q-n)\\
-(\S^{-q}K^{(q)})\left(\S^{-q-n}\frac{\dev K^{(m)}}{\dev u_n}\,\boldsymbol{\lambda\mu}(-q-n+m,-q-n)\right)=0,
\end{multline}
where we have denoted for short $\boldsymbol{\lambda\mu}(a,b)=-\boldsymbol{\lambda\mu}(b,a)=\lambda^a\mu^b-\lambda^b\mu^a$.

On the other hand, the skewsymmetric operator $K$ corresponds to the bivector
$$
P=\sum_{s>0}\frac12\int K^{(s)}\,\theta\theta_s,
$$
for which we want to compute the Schouten identity
\begin{equation}
[P,P]=2\int\frac{\delta P}{\delta\theta}\frac{\delta P}{\delta u}.
\end{equation}
We compute the two variational derivatives obtaining
\begin{align}
\frac{\delta P}{\delta\theta}&=K^{(m)}\theta_m-\left(\S^{-m}K^{(m)}\right)\theta_{-m},\\
\frac{\delta P}{\delta u}&=\left(\S^{-s}\frac{\dev K^{(m)}}{\dev u_s}\right)\theta_{-s}\theta_{m-s},
\end{align}
making
\begin{multline}
[P,P]=\frac12\int\left(K^{(m)}\left(\S^{-s}\frac{\dev K^{(q)}}{\dev u_s}\right)\theta_m\theta_{-s}\theta_{q-s}\right.\\
\left.-\;(\S^{-m}K^{(m)})\left(\S^{-s}\frac{\dev K^{(q)}}{\dev u_s}\right)\theta_{-m}\theta_{-s}\theta_{q-s}\right).
\end{multline}
We then exploit the analogue of the normalisation operator introduced by Barakat \cite{B08} which gives a standard form of elements in $\hF$. We have
\begin{equation}
\label{eq:barakat}
F=\mathcal{N}F:=\frac1p\theta\frac{\delta F}{\delta \theta}\qquad\qquad\text{for }F\in\hF^p,
\end{equation} 
where the equality holds in $\hF$, namely $F-\mathcal N F=(\S-1)G$ for $F,G\in\hF$.

We obtain $[P,P]=\int\theta T$, with
\begin{multline}\label{eq:schpf}
T=(S^{-m}K^{(m)})\left(\S^{-s-m}\frac{\dev K^{(q)}}{\dev u_s}\right)\theta_{-m-s}\theta_{-m-s+q}-\\
(\S^sK^{(m)})\frac{\dev K^{(q)}}{\dev u_s}\theta_{m+s}\theta_q+(\S^{s-q}K^{(m)})\left(\S^{-q}\frac{\dev K^{(q)}}{\dev u_s}\right)\theta_{s+m-q}\theta_{-q}-\\
K^{(m)}\left(\S^{m-s}\frac{\dev K^{(q)}}{\dev u_s}\right)\theta_{m-s}\theta_{q-s+m}+(\S^{s-m}K^{(m)})\frac{\dev K^{(q)}}{\dev u_s}\theta_{s-m}\theta_q-\\
(\S^{s-q-m}K^{(m)})\left(\S^{-q}\frac{\dev K^{(q)}}{\dev u_s}\right)\theta_{s-q-m}\theta_{-q}
\end{multline}
The vanishing of $[P,P]$ is hence equivalent to the vanishing of $T$. We observe that, after a change of names of the indices, \eqref{eq:schpf} takes the very form of \eqref{eq:pvapf} when we identify $\boldsymbol{\lambda\mu}(a,b)$ with $\theta_b\theta_a$. We illustrate the procedure for just one of the six terms -- it is straightforward to check the same for all the remaining ones.

In \eqref{eq:pvapf} we have a term of the form
\begin{equation}\label{eq:pf2}
(\S^{n-q}K^{(q)})\left(\S^{-m}\frac{\dev K^{(m)}}{\dev u_{n+m}}\right)\,\boldsymbol{\lambda\mu}(n-q,-m).
\end{equation}
In \eqref{eq:schpf}, on the other hand, we have
\begin{multline}
-(\S^{s-q-m}K^{(m)})\left(\S^{-q}\frac{\dev K^{(q)}}{\dev u_s}\right)\theta_{s-q-m}\theta_{-q}=\\
(\S^{s-q-m}K^{(m)})\left(\S^{-q}\frac{\dev K^{(q)}}{\dev u_s}\right)\theta_{-q}\theta_{s-q-m}.
\end{multline}
A change of indices $(t\to m,n\to s-q,m\to q)$ in \eqref{eq:pf2} gives
\begin{equation}
(\S^{s-q-m}K^{(m)})\left(\S^{-q}\frac{\dev K^{(q)}}{\dev u_{s}}\right)\,\boldsymbol{\lambda\mu}(s-q-m,-q),
\end{equation}
namely the corresponding term in $T$ after the aforementioned identification.
\subsection{Compatible pairs of multiplicative PVA}\label{app:compPairs}
Two $\lambda$ brackets $\{\cdot_{\lambda}\cdot\}_{1,2}$ of multiplicative PVAs are said to be \emph{compatible} if the bracket $\{\cdot_{\lambda}\cdot\}_1+\alpha \{\cdot_{\lambda}\cdot\}_2$ is the $\lambda$ bracket of a multiplicative PVA for all the values of $\alpha\in\C$.

In \cite{dSKVW18}, the authors provide the classification of all the scalar multiplicative PVAs up to the order $(-5,5)$ (we say that a scalar multiplicative $\lambda$ bracket is of order $(M,N)$ if it is of the form $\{u_{\lambda}u\}=\sum_{s=M}^N a^{(s)}\lambda^s$, namely if the corresponding difference operator is of the same order). The compatible pairs among their list are the following \cite[Theorem 2.5, 8.1, 9.1]{dSKVW18}:
\begin{enumerate}
	\item For a function $g(u)$ and $c_2,\ldots,c_5$ arbitrary constants, the pair $$
	\{u_\lambda u\}_{1,g}:=g(u)g(u_1)\lambda-g(u)g(u_{-1})\lambda^{-1}$$
	and 
	\begin{equation}\label{eq:Kgen}
	\{u_\lambda u\}_g=\sum_{k=2}^5c_k\{u_\lambda u\}_{k,g}:=\sum_{k=2}^5 c_k\left(g(u)g(u_k)\lambda^k-g(u)g(u_{-k})\lambda^{-k}\right).
	\end{equation}
	\item For a nonzero function $g(u)$, $\{u_\lambda u\}_{1,g}$ as above and
	\begin{align}\label{eq:K2g}
	\{u_\lambda u\}^\sim_{2,g}&:=g(u)g(u_2)e^{F(u_1)}\lambda^2+g(u)g(u_1)\left(e^{F(u)}+e^{F(u_1)}\right)\lambda\\
	&\quad-\S^{-1}g(u)g(u_1)\left(e^{F(u)}+e^{F(u_1)}\right)\lambda^{-1}-\S^{-2}g(u)g(u_2)e^{F(u_1)}\lambda^{-2},
	\end{align}
	where $F(u)=\int^u\frac1g$.
	\item For a nonzero function $g(u)$, $\{u_\lambda u\}_{1,g}$ as above and
	\begin{align}
	\{u_\lambda u\}^\sim_{3,g}&:=g(u)g(u_3)e^{F(u_1)+\I F(u_2)}\lambda^3+\\
	&\quad+\I g(u)g(u_2)\left(e^{F(u)+\I F(u_1)}-e^{F(u_1)+\I F(u_2)}\right)\lambda^2+\\
	&\quad+g(u)g(u_1)e^{F(u)+\I F(u_1)}\lambda-\S^{-1}g(u)g(u_1)e^{F(u)+\I F(u_1)}\lambda^{-1}\\
	&\quad-\S^{-2}\I g(u)g(u_2)\left(e^{F(u)+\I F(u_1)}-e^{F(u_1)+\I F(u_2)}\right)\lambda^{-2}\\ \label{eq:K3g}
	&\quad-\S^{-3}g(u)g(u_3)e^{F(u_1)+\I F(u_2)}\lambda^{-3}.
	\end{align}
	\item For a nonzero function $g(u)$, $\{u_\lambda u\}^\sim_{2,g}$ as above and the $(-4,4)$ order bracket
	\begin{align}
	\{u_{\lambda}u\}^{(iv)}_g&=g(u)g(u_4)e^{F(u_1)-F(u_2)+F(u_3)}\lambda^4-\S^{-4}g(u)g(u_4)e^{F(u_1)-F(u_2)+F(u_3)}\lambda^{-4}\\\label{eq:Qg}
	&\quad+g(u)g(u_3)\left(e^{F(u)-F(u_1)+F(u_2)}+e^{F(u_1)-F(u_2)+F(u_3)}\right)\lambda^3+\\
	&\quad-\S^{-3}g(u)g(u_3)\left(e^{F(u)-F(u_1)+F(u_2)}+e^{F(u_1)-F(u_2)+F(u_3)}\right)\lambda^{-3}\\
	&\quad +g(u)g(u_2)e^{F(u)-F(u_1)+F(u_2)}\lambda^2-\S^{-2}g(u)g(u_2)e^{F(u)-F(u_1)+F(u_2)}\lambda^{-2}.
	\end{align}
	\item For a nonzero function $g(u)$, $\{u_\lambda u\}_{1,g}+\{u_\lambda u\}_{2,g}$ as defined above and the following $(-5,5)$ order bracket
	\begin{align}\label{eq:QQg}
	\{u_\lambda u\}^{(v)}_g&=g(u)g(u_5)e^{\varepsilon F(u_2)+F(u_3)}\lambda^5-\S^{-5}g(u)g(u_5)e^{\varepsilon F(u_2)+F(u_3)}\lambda^{-5}\\
	&\quad -g(u)g(u_4)\left(\varepsilon e^{\varepsilon F(u_1)+F(u_2)}+\varepsilon^{-1} e^{\varepsilon F(u_2)+F(u_3)}\right)\lambda^4\\
	&\quad+\S^{-4}g(u)g(u_4)\left(\varepsilon e^{\varepsilon F(u_1)+F(u_2)}+\varepsilon^{-1} e^{\varepsilon F(u_2)+F(u_3)}\right)\lambda^{-4}\\
	&\quad+g(u)g(u_3)\left(\varepsilon^{-1}e^{\varepsilon F(u)+F(u_1)}+e^{\varepsilon F(u_1)+F(u_2)}+\varepsilon e^{\varepsilon F(u_2)+F(u_3)}\right)\lambda^3\\
	&\quad-\S^{-3}g(u)g(u_3)\left(\varepsilon^{-1}e^{\varepsilon F(u)+F(u_1)}+e^{\varepsilon F(u_1)+F(u_2)}+\varepsilon e^{\varepsilon F(u_2)+F(u_3)}\right)\lambda^{-3}\\
	&\quad-g(u)g(u_2)\left(\varepsilon e^{\varepsilon F(u)+F(u_1)}+\varepsilon^{-1} e^{\varepsilon F(u_1)+F(u_2)}\right)\lambda^{2}\\
	&\quad\S^{-2}g(u)g(u_2)\left(\varepsilon e^{\varepsilon F(u)+F(u_1)}+\varepsilon^{-1} e^{\varepsilon F(u_1)+F(u_2)}\right)\lambda^{-2}\\
	&\quad+g(u)g(u_1)e^{\varepsilon F(u)+F(u_1)}\lambda-\S^{-1}g(u)g(u_1)e^{\varepsilon F(u)+F(u_1)}\lambda^{-1},
	\end{align}
	where we denote $\varepsilon$ a primitive 3rd root of 1.
\end{enumerate}
Each of these brackets defines, equivalently, a scalar Hamiltonian operator, and compatible brackets define biHamiltonian pairs.

\section{Change of coordinates on $\tilde{K}_{2,g}$}\label{app:chC}
In Section \ref{sec:applic} we obvserved that the dependancy on the function of single variable $g(u)$ in the classification presented in \cite{dSKVW18} -- and in particular for the compatible operators listed in Appendix \ref{app:compPairs} -- can be removed by a change of coordinates
\begin{equation}
u\mapsto v=\int \frac{1}{g(u)}.
\end{equation}

In this paragraph we demonstrate the effect of this change of coordinates on the Hamiltonian operator $\tilde{K}_{2,g}$ defined in \eqref{eq:K2g}. Its associated $\lambda$ bracket is
\begin{align}
\{u_{\lambda}u\}&=g(u)g(u_2)e^{\int^{u_1}\frac{1}{g}}\lambda^2+g(u)g(u_1)\left(e^{\int^{u}\frac{1}{g}}+e^{\int^{u_1}\frac{1}{g}}\right)\lambda\\
&\quad-g(u)g(u_{-1})\left(e^{\int^{u}\frac{1}{g}}+e^{\int^{u_{-1}}\frac{1}{g}}\right)\lambda^{-1}-g(u)g(u_{-2})e^{\int^{u_{-1}}\frac{1}{g}}\lambda^{-2}.
\end{align}

We compute the bracket in the new coordinate
$$
v=\int\frac{1}{g(u)}
$$
using the master formula \eqref{eq:master}. By the skewsymmetry of the bracket, it is sufficient to perform the change of coordinate for the terms in the positive degrees for $\lambda$, while the negative degree will follow according to the general rule \eqref{eq:skewbracket-det}. 

For simplicity, we compute the transformation for each summand separately, according to the order of $\lambda$: $\{u_{\lambda}u\}=\{u_{\lambda}u\}^{(2)}+\{u_{\lambda}u\}^{(1)}-\{u_{\lambda}u\}^{(-1)}-\{u_{\lambda}u\}^{(-2)}$. We have
\begin{multline}
\{v_{\lambda}v\}^{(2)}=\frac{1}{g(u)} g(u)g(u_2)e^{\int^{u_1}\frac1g}\left(\S^2\frac{1}{g(u)}\right)\lambda^2=\\
\frac{1}{g(u)}g(u)g(u_2)\frac{1}{g(u_2)}e^{\int^{u_1}\frac1g}\lambda^2=e^{\int^{u_1}\frac1g}\lambda^2=e^{v_1}\lambda^2
\end{multline}
and
\begin{multline}
\{v_{\lambda}v\}^{(1)}=\frac{1}{g(u)} g(u)g(u_1)\left(e^{\int^{u}\frac1g}+e^{\int^{u_1}\frac1g}\right)\left(\S\frac{1}{g(u)}\right)\lambda=\\
\frac{1}{g(u)}g(u)g(u_1)\frac{1}{g(u_1)}\left(e^{\int^{u}\frac1g}+e^{\int^{u_1}\frac1g}\right)\lambda\\
=\left(e^{\int^{u}\frac1g}+e^{\int^{u_1}\frac1g}\right)\lambda=\left(e^v+e^{v_1}\right)\lambda.
\end{multline}
$\{u_{\lambda}u\}^{(-1)}$ and $\{u_{\lambda}u\}^{(-2)}$ are then obtained by the skewsymmetry property. Notice that the result we obtained is the same as substituting $g(u)=1$ in the definition of the bracket.

\section*{Acknowledgements}
The paper is supported by the EPSRC grant
EP/P012698/1. Both authors gratefully acknowledge the financial support.

\bibliographystyle{plain}
\bibliography{biblio-diffcoho}
\end{document}